\documentclass[10pt,letterpaper,twocolumn]{article}

\usepackage[
  letterpaper,
  left=0.68in,
  right=0.68in,
  top=0.75in,
  bottom=1in,
  columnsep=0.24in
]{geometry}
\usepackage{lmodern}
\usepackage[T1]{fontenc}
\usepackage{microtype}
\usepackage{xspace}
\usepackage{booktabs}
\usepackage{array}
\usepackage{multirow}
\usepackage{graphicx}
\usepackage{hyperref}
\usepackage{xcolor}
\usepackage[font=footnotesize,labelfont=bf]{caption}
\usepackage[compact]{titlesec}
\usepackage{placeins}
\usepackage{amsmath}
\usepackage{amssymb}
\usepackage{amsthm}
\newtheorem{theorem}{Theorem}

\hypersetup{colorlinks=true, linkcolor=black, citecolor=black, urlcolor=blue}
\renewcommand{\arraystretch}{0.96}
\titlespacing*{\section}{0pt}{1.1ex plus .3ex minus .2ex}{0.55ex plus .15ex}
\titlespacing*{\subsection}{0pt}{0.85ex plus .25ex minus .15ex}{0.35ex plus .1ex}
\titlespacing*{\paragraph}{0pt}{0.45ex plus .15ex minus .1ex}{0.55em}

\newcommand{\method}{\textsc{MemLineage}\xspace}
\newcommand{\attackSleeper}{\textit{sleeper-via-derivation}\xspace}
\newcommand{\attackPoison}{\textit{AgentPoison-style}\xspace}
\newcommand{\attackGraft}{\textit{MemoryGraft-style}\xspace}

\newcommand{\trustTrusted}{\textsc{Trusted}\xspace}
\newcommand{\trustDerivedTrusted}{\textsc{Derived-Trusted}\xspace}
\newcommand{\trustDerivedUntrusted}{\textsc{Derived-Untrusted}\xspace}
\newcommand{\trustExternal}{\textsc{External}\xspace}

\newcommand{\algoCoarse}{\textsc{Coarse}\xspace}
\newcommand{\algoLmEval}{\textsc{LmSelfEval}\xspace}
\newcommand{\algoAttn}{\textsc{AttnAttr}\xspace}

\newcommand{\up}{$\uparrow$}
\newcommand{\down}{$\downarrow$}

\title{\method{}: Lineage-Guided Enforcement for LLM Agent Memory}
\author{%
Ciyan Ouyang\\
State Key Laboratory of Cyberspace Security Defense\\
Institute of Information Engineering, CAS\\
Beijing, China
\and
Rui Hou\thanks{Corresponding author.}\\
State Key Laboratory of Cyberspace Security Defense\\
Institute of Information Engineering, CAS\\
Beijing, China
}
\date{}

\begin{document}
\maketitle

\begin{abstract}
We introduce \method{}, a defence for LLM agent memory that
attaches both cryptographic provenance and LLM-mediated derivation
lineage to every entry. Recent and concurrent work shows that
untrusted content can be written into persistent agent state and
re-enter later sessions as an instruction; the remaining systems
question is how to preserve useful memory recall while preventing
such state from justifying sensitive actions. \method{} treats this
as a chain-of-custody problem rather than a filtering problem. It is
a six-module design around an RFC-6962 Merkle log over per-principal
Ed25519-signed entries: a weighted derivation DAG records which
retrieved entries influenced each new memory, and a
max-of-strong-edges propagation rule makes Untrusted-Path
Persistence hold for any chain whose attribution edges remain above
threshold. The sensitive-action gate then refuses dispatches whose
active justification descends from an external ancestor, while still
    allowing benign recall. We evaluate three defence cells against
    three memory-poisoning workloads on a deterministic
    mechanism-isolation harness; \method{} is the only configuration in
    that harness that drives all three columns to zero ASR, while
    sub-millisecond per-operation overhead keeps it well below the
    noise floor of any LLM call. A Codex-backed AgentDojo bridge
    further separates strong-model behaviour from defence-layer
    behaviour: under an intentionally vulnerable tool-output profile,
    no-defence and signature-only baselines fail on all six banking
    pairs, while all \method{} rows reduce strict AgentDojo ASR to
    zero. The core deterministic artifacts (ASR matrix,
    $\tau \times K$ lineage ablation, utility baseline, and recovery
    POCs) are byte-equal CI-verified; hosted-model AgentDojo and
    live-model sweeps are recorded as auditable logs rather than
    byte-pinned artifacts.
\end{abstract}

\section{Introduction}
\label{sec:intro}

Modern LLM agents persist memory across user sessions: chat logs,
ingested documents, tool returns, and the agent's own derived
observations all accumulate in a long-running store the next user
turn will retrieve from. This memory is the agent's \emph{capability}
substrate -- it is what allows multi-step workflows to compose -- and
it is increasingly the agent's attack surface. Indirect prompt
injection through the active context~\cite{greshake2023indirect} has
been studied at length; the published memory-poisoning literature
adds three named families:
AgentPoison~\cite{chen2024agentpoison} plants trigger-bound
backdoors that surface on a matching query;
MINJA~\cite{dong2025minja} injects payloads through normal
query-time interactions; and
MemoryGraft~\cite{srivastava2025memorygraft} commits poisoned
experience entries that persist across sessions. Zombie
Agents~\cite{yang2026zombieagents} and a concurrent defence
evaluation~\cite{leong2026defenseeffectiveness} sharpen this into a
two-phase persistent-memory threat: attacker-controlled content is
stored through the agent's normal memory update path and later
re-enters a different session as instruction. A 2025 SoK
catalogues twenty-three IPI-centric defences against
these~\cite{ji2025ipisok}. The attack class is therefore no longer
speculative. The open systems problem is how to preserve useful
long-term memory while preventing untrusted memory ancestry from
authorising sensitive actions.

\paragraph{The enforcement gap.}
The hard case for a memory defence is not only a malicious string
stored verbatim. A patient adversary can plant \emph{ingredients} in
untrusted sources, then wait for the agent's normal retrieval and
summarisation loop to emit a fresh derived entry that the agent
commits under its own writer principal. The derived entry is, by
construction, an authentic output of a registered principal,
indistinguishable to a signature layer from a benign summary of the
same context. Many sessions later, an unrelated query retrieves the
derived entry and the laundered payload triggers a sensitive action.
This chain-of-custody failure exposes three capability gaps:

\begin{itemize}
  \item Signature-only memory layers, in the spirit of in-toto-style
    supply-chain integrity~\cite{torres-arias2019intoto}, verify
    writer identity but do not recover \emph{what} the entry was
    derived from.
  \item IFC-based agent planners such as
    Fides~\cite{costa2025fides} enforce confidentiality and integrity
    labels at planning time but do not persist labels across the
    LLM derivation step that commits the laundered entry.
  \item Retrieval-stage filters such as the
    RAGPart~/~RAGMask family~\cite{pathmanathan2025ragdef} operate
    on the recall surface but cannot tell a benign summary from a
    laundered payload arriving through an authentic principal.
  \item Coarse memory-layer defences that remove recall or attenuate
    all capabilities after exposed reads~\cite{leong2026defenseeffectiveness,zha2026agentworms}
    can be effective, but give up the fine-grained distinction
    between useful derived memories and memories whose sensitive
    use depends on untrusted ancestry.
\end{itemize}

\paragraph{Approach.}
We introduce \method{}, which attaches \emph{cryptographic provenance
and lineage} to every memory entry. Provenance is per-entry Ed25519
signature anchored in an RFC~6962 Merkle log~\cite{laurie2013rfc6962};
lineage is a directed acyclic graph over the memory store whose
edges record LLM-mediated derivation, weighted by an attribution
algorithm. A \emph{max-of-strong-edges} propagation rule (\S\ref{sec:design-lineage})
makes Untrusted-Path Persistence hold: any path from an
\trustExternal{} ancestor whose every edge is strong forces the
chain tip to inherit the untrusted label, which the
sensitive-action gate refuses. The capability the attribution
algorithm has to deliver, formally, is the recall $r_K$ of the
strong-edge predicate at chain length $K$; we ablate three
algorithms (uniform-weight, LM-self-eval, white-box attention) and
sweep $\tau \times K$ to map the trade-off.

\paragraph{Evaluation summary.}
We instantiate \method{} as a Python library that integrates with
LangGraph and evaluate it on a deterministic harness that pins
attacker behaviour and avoids real LLM API calls. Three defence
cells are compared on three attack families (Table~\ref{tab:asr}):
\texttt{no\_defense} fails all three; a \texttt{sig\_only\_baseline}
(per-entry Ed25519 plus Merkle inclusion, no lineage) closes the
\attackPoison{} column but fails on \attackGraft{} and
\attackSleeper{}; \method{} drives all three columns to zero on this
mechanism-isolation harness. The
$\tau \times K$ ablation under a degrading judge schedule
(Table~\ref{tab:rk-lmeval}) exposes the boundary that a
single-shot $\tau$ measurement would miss: under
$w_0 = 0.9$, $d = 0.7$, the safe $\tau$ ceiling shrinks with $K$
(the deepest edge weight $w_K \!=\! 0.9 \cdot 0.7^{K-1}$ sets the
binding constraint), so a deployment allowing $K \!=\! 5$
derivation hops needs $\tau \le 0.10$ rather than the $\tau \!\le\! 0.30$
a $K\!=\!1$ measurement would suggest.
Per-operation overhead (Table~\ref{tab:perf}) is sub-millisecond on
the hot path -- well below the noise floor of any LLM call. We also
ship a Codex-backed AgentDojo bridge: the default DirectAttack sweep
shows that modern instruction-tuned models can make conservative
no-defence ASR look low (1/6), while an explicitly vulnerable
tool-output profile drives both no-defence baselines to 6/6 and all
\method{} rows to 0/6 strict AgentDojo ASR.

\paragraph{Contributions.}
\begin{itemize}
  \item \textbf{A lineage-stress memory-laundering workload}
    (\S\ref{sec:atk-sleeper},
    Figure~\ref{fig:sleeper-chain}):
    a clean three-stage workload that operationalises the hard case
    surfaced by persistent-memory attacks: untrusted content is
    transformed through LLM-mediated derivation into an authentic
    agent-written entry. We use it as a separator between
    signature-only, coarse-taint, and lineage-aware memory defences,
    not as a claim that cross-session persistence itself is novel.
  \item \textbf{The \method{} design} (\S\ref{sec:design},
    Figure~\ref{fig:architecture}):
    six modules around a single memory store, with the
    max-of-strong-edges propagation rule and an
    Untrusted-Path Persistence theorem (Theorem~\ref{thm:soundness})
    formalising what the lineage layer actually buys. M6 further
    checks tool arguments against per-tool authority rules, which
    supports utility-preserving repair when trusted provenance supplies
    authorized replacement values for attacker-controlled parameters.
    The four
    capability dimensions that define the system's operating point
    are summarised in Figure~\ref{fig:capability-matrix}
    (\S\ref{sec:related-position}).
  \item \textbf{Three attribution algorithms} (\S\ref{sec:design-lineage}):
    uniform-weight, secondary-LLM judge with prompt-injection
    hardening, and white-box attention. We quantify $K$-step
    attribution recall $r_K$ and its interaction with the lineage
    threshold $\tau$ (\S\ref{sec:eval-rk}).
  \item \textbf{An end-to-end evaluation harness} that produces
    reproducible artifacts under \texttt{paper/data/}. The ASR matrix,
    the two-session RAG-to-memory workflow, the $\tau \times K$
    ablation, adaptive mitigation cells, recovery POCs, and the
    utility baseline are byte-equal CI-verified
    (\texttt{asr\_matrix\_v1.csv},
    \texttt{rag\_memory\_e2e\_v1.csv},
    \texttt{r\_k\_ablation\_v1.csv},
    \texttt{adaptive\_mitigations\_v1.csv},
    \texttt{coarse\_authority\_contrast\_v1.csv},
    \texttt{coarse\_authority\_matrix\_v1.csv},
    \texttt{authority\_source\_manifest\_v1.csv},
    \texttt{authority\_repair\_poc\_v1.csv},
    \texttt{rag\_memory\_llm\_summary\_v1.csv},
    \texttt{utility\_baseline\_v1.csv}). Hosted-model AgentDojo and
    live-model sweep artifacts are preserved as auditable logs rather
    than byte-pinned tests, and the
    per-operation micro-bench (\texttt{perf\_baseline\_v1.csv}) is
    regenerated because its numbers are wall-clock-dependent.
\end{itemize}

\paragraph{Roadmap.}
\S\ref{sec:threat} fixes the threat model and trust assumptions.
\S\ref{sec:design} presents the six-module design and the
propagation rule. \S\ref{sec:impl} discusses the implementation
boundary, including the JSON-envelope hardening that prevents the
LLM judge from being prompt-injected.
\S\ref{sec:attacks} specifies the three attack families.
\S\ref{sec:evaluation} reports the empirical results.
\S\ref{sec:related} positions \method{} against published
defences and concurrent preprints (Fides, RAGPart/RAGMask,
NeuroTaint, A-MemGuard, Memory Sandbox, RTW-A, and the IPI-defence
landscape catalogued in~\cite{ji2025ipisok}).
\S\ref{sec:discussion} treats limitations -- the LLM-inference
trust assumption, the white-box attention requirement of the
strongest attribution algorithm, and the deferred HMAC fast-path.

\section{Threat Model}
\label{sec:threat}

This section fixes the system, the attacker, and the trust boundaries
that the rest of the paper operates under. The five-front survey of
named defences and named attacks is deferred to \S\ref{sec:related};
here we describe only the model that those defences and attacks share.

\subsection{System Model}
\label{sec:threat-system}

An \emph{LLM agent} is a host process that drives a chat-completion
model in a loop, calling tools and consulting memory between turns.
The agent has three durable surfaces relevant to provenance:

\begin{itemize}
  \item A \emph{memory store} $M$: a content-addressed set of entries
    written by either the host (on user request) or the agent itself
    (when the agent decides a derived observation is worth keeping).
    Each entry carries a payload, a creation timestamp, and any
    metadata the host or agent attaches.
  \item A \emph{retrieval surface} $R(q) \subseteq M$: an
    embedding-based or BM25-based recall step that selects a
    bounded subset of $M$ to splice into the next prompt. We do
    not require any specific retrieval algorithm; the threat model
    only assumes that the surface returns entries authentic to $M$
    (see \S\ref{sec:threat-trust}).
  \item A \emph{tool surface} $T$: a fixed set of functions the agent
    may invoke between turns. A subset $T_{\text{sens}} \subseteq T$
    is policy-sensitive (file-system writes, financial transfers,
    outbound HTTP, code execution, etc.); the security goal is to
    prevent unauthorised dispatch into $T_{\text{sens}}$.
\end{itemize}

The host process holds long-term key material on behalf of one or
more \emph{principals} (the user, the agent itself, and any upstream
agent or tool that is permitted to write into $M$). Principal keys
are used to sign entries written into $M$; key custody is part of
the trusted base (\S\ref{sec:threat-trust}).

\subsection{Attacker Model}
\label{sec:threat-attacker}

We assume a \emph{white-box} adversary with full knowledge of
\method{}'s design, metadata schema, lineage propagation rule, and
sensitive-action policy. The adversary is omniscient about the
defence; they are constrained only by the capability boundaries
listed below.

\paragraph{Adversary capabilities.}
The adversary controls at least one untrusted data source reachable
by the agent and may exercise any of:

\begin{enumerate}
  \item \textbf{Untrusted-source authorship.} The adversary controls
    one or more of: web pages the agent fetches, documents the agent
    ingests, output channels of an external tool the agent invokes, or
    messages from an upstream agent that is permitted to write into
    $M$. This is the standard indirect-prompt-injection
    threat~\cite{greshake2023indirect}.
  \item \textbf{Memory write through legitimate paths.} The adversary
    can inject entries into $M$ \emph{only} along paths the agent or
    host already permits: by being read from a controlled web page,
    by being returned from a controlled tool, by being delivered as a
    message from an upstream agent. The adversary cannot bypass the
    write surface to insert raw entries with arbitrary metadata.
  \item \textbf{Cross-session persistence.} Payloads injected in
    session~$s_i$ may remain in $M$ and re-emerge through retrieval
    in any later session~$s_j$, $j > i$. We refer to this as a
    \emph{sleeper} pattern, modelled directly on Zombie
    Agents~\cite{yang2026zombieagents},
    AgentPoison-style trigger-bound backdoors~\cite{chen2024agentpoison}
    and MemoryGraft-style retrieval poisoning~\cite{srivastava2025memorygraft}.
  \item \textbf{LLM-mediated laundering.} The adversary's payload may
    enter $M$ \emph{indirectly}: an untrusted source gets retrieved
    into context, the LLM emits a derived summary or rephrasing, and
    the agent commits the derived content as a fresh entry. This
    laundering step is the lineage-stress case the present paper uses
    (\S\ref{sec:attacks}).
  \item \textbf{Adaptive querying} (deferred). The adversary may, in
    principle, issue probe queries against the deployed system to
    estimate its lineage threshold and policy verdicts and tune
    subsequent payloads. The deterministic harness in
    \S\ref{sec:evaluation} fixes a non-adaptive payload sequence per
    cell so that the ASR numbers are reproducible without API
    calls; we discuss adaptive evaluation as future work in
    \S\ref{sec:discussion}.
\end{enumerate}

\paragraph{Adversary non-capabilities.}
The following are explicitly outside the adversary's reach in this
paper:

\begin{enumerate}
  \item \textbf{Host process integrity.} The adversary cannot modify
    the host binary, the \method{} library code, or any in-process
    state outside $M$.
  \item \textbf{Principal-key secrecy.} The adversary cannot read
    private keys held by registered principals (the user, the agent,
    upstream principals). Equivalently, the adversary cannot forge a
    signature under a principal whose key they do not hold.
  \item \textbf{LLM-inference manipulation.} The adversary cannot
    rewrite the LLM weights or inject an external compute path that
    bypasses the agent's normal inference pipeline. Their only
    influence on inference is through the contents of $M$ and any
    untrusted source that is retrieved into context.
  \item \textbf{Attestation forgery.} If the deployment opts in to
    TEE-bound key custody (an extension we discuss in
    \S\ref{sec:discussion}, not part of the main evaluation), the
    adversary cannot forge a hardware attestation report.
\end{enumerate}

\subsection{Trust Assumptions}
\label{sec:threat-trust}

We use a three-tier source classification consistent with
\cite{chen2024agentpoison,srivastava2025memorygraft,greshake2023indirect}
and aligned with the broader IPI-defence taxonomy
in~\cite{ji2025ipisok}:

\begin{itemize}
  \item \trustTrusted{}: user prompts; entries written by the host on
    explicit user authorisation; principal-key infrastructure
    (private keys, public-key registry, Merkle log root anchor).
  \item \emph{Semi-trusted}: tool return values from sandboxed or
    audited tools, and entries derived from them through a documented
    policy. Semi-trusted content can be cryptographically authenticated
    as having come from the named tool, but its \emph{contents} can
    still be attacker-influenced if the tool consumes adversary-controlled
    input.
  \item \trustExternal{}: arbitrary fetched web content, third-party
    document corpora, messages from upstream agents whose host trust
    we do not control, and any payload reached transitively through
    those.
\end{itemize}

These classes induce two derived labels that \method{} reasons about
explicitly: \trustDerivedTrusted{} (an entry derived by the LLM from
parents that all evaluate above the lineage threshold $\tau$) and
\trustDerivedUntrusted{} (the dual case). The propagation rule that
defines them is given in \S\ref{sec:design}; the empirical sweep over
$\tau$ and chain length $K$ is in \S\ref{sec:eval-rk}.

\subsection{Attack Pattern Taxonomy}
\label{sec:threat-patterns}

The threat surface partitions into four families. \method{} addresses
the last three; the first is included only to fix the boundary.

\begin{enumerate}
  \item \emph{Single-turn direct injection}~\cite{greshake2023indirect}:
    a malicious string in the current turn's untrusted context steers
    the model output before any memory write occurs. This family is
    well covered by detection-, sanitisation-, and IFC-based
    defences~\cite{ji2025ipisok}; we treat it as orthogonal and out
    of scope.
  \item \emph{Memory-backdoor poisoning}~\cite{chen2024agentpoison}: the
    adversary plants trigger-bound entries that, once retrieved by a
    matching query, steer the agent toward an attacker-chosen action.
    The hard sub-case is the \emph{sleeper}: trigger-bound entries
    that lie dormant across many sessions before activation.
  \item \emph{Query-injection / retrieval poisoning}~\cite{dong2025minja,srivastava2025memorygraft}:
    the adversary plants entries chosen so that benign user queries
    surface them through the standard retrieval ranking, and the
    surfaced content carries injection. MemoryGraft additionally
    models persistent compromise via poisoned experience entries.
  \item \emph{Sleeper-via-derivation} (\S\ref{sec:attacks}):
    the adversary's payload never enters $M$ directly. The adversary
    plants only \emph{ingredients} in untrusted sources; the LLM is
    coerced, through normal retrieval and summarisation, to
    \emph{emit} the payload as a fresh derived entry that the agent
    commits. The derived entry is, by construction, an authentic
    output of the agent principal -- which is precisely why
    signature-only defences cannot rule it out. This family is the
    benchmark stressor the paper uses to test whether a defence
    preserves chain of custody through LLM-mediated derivation.
\end{enumerate}

\subsection{Security Goals}
\label{sec:threat-goals}

\method{}'s evaluation in \S\ref{sec:evaluation} is judged against
four informal goals; \S\ref{sec:design} formalises the propagation
rule that satisfies them.

\begin{description}
  \item[G1: Source-bound integrity.] No entry that is recoverable from
    $M$ can have been authored by a principal whose key is not in the
    public-key registry. Equivalently, the adversary cannot inject
    raw entries that pass verification.
  \item[G2: Lineage-bound trust.] An entry whose lineage chain
    contains an \trustExternal{} ancestor cannot present as
    \trustDerivedTrusted{} unless every step from the ancestor passes
    the lineage threshold $\tau$.
  \item[G3: Cross-session persistence of provenance.] Provenance
    metadata is content-addressed and signed; it is invariant under
    serialisation, transport, and cross-session retrieval. A
    derived-untrusted chain remains derived-untrusted across
    arbitrarily many session boundaries.
  \item[G4: Sensitive-action gating.] No call into $T_{\text{sens}}$
    is dispatched whose justifying memory contains an entry that
    fails G2. This is the operational invariant the policy gate
    enforces (\S\ref{sec:design}).
\end{description}

The signature-only baseline (\S\ref{sec:eval-asr}) closes G1 alone.
G2 and G4 require the lineage layer; G3 requires both. The empirical
claim of the paper is that the four goals are jointly enforceable at
sub-millisecond per-operation cost (\S\ref{sec:eval-perf}) with no
loss of utility on the deterministic harness (\S\ref{sec:eval-asr}).

\subsection{Out of Scope}
\label{sec:threat-oos}

The following adversaries and surfaces are explicitly outside the
present paper:

\begin{itemize}
  \item Single-turn direct prompt injection with no memory write
    (covered by~\cite{greshake2023indirect,ji2025ipisok}).
  \item Model-weight-level backdoors and supply-chain compromise of
    the inference engine; we assume the LLM weights are trusted.
  \item Host-process compromise (kernel exploits, container escapes,
    debugger attach). The host trusted base is responsible for
    keeping principal keys and the Merkle log anchor confidential.
  \item Manipulation of retrieval inference itself beyond the
    contents of $M$. Active learning attacks against the retriever's
    ranking model are a separate threat class.
  \item Defences against the LLM revealing private memory contents to
    the adversary (orthogonal confidentiality goals; e.g.,
    \cite{costa2025fides} addresses these via IFC labels and is
    largely complementary to provenance).
\end{itemize}

The combination of (i) cryptographic label integrity,
(ii) cross-session persistence, (iii) coverage of agent-derived
entries, and (iv) defence against LLM-mediated laundering is what
\method{} contributes; \S\ref{sec:related} positions this
entry-level memory primitive against coarse recall removal,
temporal re-entry control, semantic taint, and experience-driven
memory defences.

\section{Design}
\label{sec:design}

\method{} attaches cryptographic provenance and per-derivation lineage
to every memory entry, and refuses sensitive actions whose justifying
context contains an entry whose chain of custody crosses an
\trustExternal{} ancestor through a too-strong edge. This section
fixes the architecture (\S\ref{sec:design-arch}), the metadata schema
(\S\ref{sec:design-meta}), the cryptographic binding
(\S\ref{sec:design-crypto}), the append-only log
(\S\ref{sec:design-log}), the lineage DAG and its propagation rule
(\S\ref{sec:design-lineage}), the verifier-aware retrieval surface
(\S\ref{sec:design-retrieval}), and the sensitive-action gate
(\S\ref{sec:design-policy}). The attack used to evaluate the design
end-to-end (\attackSleeper{}) is described separately in
\S\ref{sec:attacks}.

\subsection{Architecture}
\label{sec:design-arch}

\method{} is six modules sitting around a single memory store $M$;
Figure~\ref{fig:architecture} renders the dataflow.

\begin{figure}[t]
  \centering
  \includegraphics[width=\columnwidth]{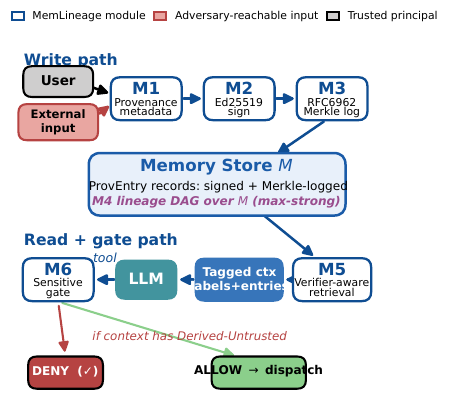}
  \caption{\method{} architecture: six modules around a single
    memory store $M$. The write path (top) commits an entry
    through M1 metadata, M2 per-principal Ed25519 signing, and M3
    RFC-6962 Merkle logging; the read + gate path (bottom)
    retrieves through M5, attaches trust labels via M4 lineage
    propagation, and refuses sensitive tool calls at M6 whose
    justifying context carries a \trustDerivedUntrusted{} ancestor.
    Adversary-reachable inputs (red) are the only path through
    which untrusted content reaches $M$; the rest of the pipeline
    is host-trusted (\S\ref{sec:threat}).}
  \label{fig:architecture}
\end{figure}

\begin{description}
  \item[M1 -- Provenance metadata.] A canonical schema attached to
    every entry that records the writer principal, the source-data
    hash, the write-time context hash, the trust level, and the
    parent-edge list. Defined in \S\ref{sec:design-meta}.
  \item[M2 -- Cryptographic binding.] Ed25519 signature over the
    canonical encoding of every entry, keyed per principal.
    \S\ref{sec:design-crypto}.
  \item[M3 -- Append-only log.] An RFC~6962-style Merkle
    log~\cite{laurie2013rfc6962} of entry hashes with explicit
    tombstone semantics. \S\ref{sec:design-log}.
  \item[M4 -- Lineage DAG.] A directed acyclic graph over $M$ whose
    edges record LLM-mediated derivation, weighted by an attribution
    algorithm. The propagation rule on this DAG is what closes the
    \attackSleeper{} column. \S\ref{sec:design-lineage}.
  \item[M5 -- Verifier-aware retrieval.] A retrieval-time hook that
    verifies signatures, looks up trust levels, and renders a
    \emph{tagged} context segment for each retrieved entry.
    \S\ref{sec:design-retrieval}.
  \item[M6 -- Sensitive-action gate.] A policy gate aligned with
    Progent~\cite{shi2025progent} that examines the trust labels in
    the active prompt before dispatching any call into
    $T_{\text{sens}}$ (\S\ref{sec:threat-system}).
    \S\ref{sec:design-policy}.
\end{description}

The write path is M1~$\to$~M2~$\to$~M3 (with M4 attaching parent
edges); the read path is M3~$\to$~M5~$\to$~prompt, with M6 gating any
tool dispatch produced by the resulting LLM turn. The single
threaded invariant the rest of the section preserves is that no entry
ever leaves M5 without a verified signature and a current trust label.

\subsection{Provenance Metadata (M1)}
\label{sec:design-meta}

Each entry is a record of the form
\begin{equation*}
  \begin{aligned}
    \texttt{ProvEntry}=\langle&
      \mathit{eid}, \mathit{content}, \mathit{writer},
      h_{\text{src}}, h_{\text{ctx}},\\
    & t, \mathit{trust}, P, W, \mathit{ts},
      \mathit{nonce}, \sigma\rangle .
  \end{aligned}
\end{equation*}

where
$\mathit{eid}$ is a UUIDv7 entry identifier;
$\mathit{content}$ is the payload bytes (optionally encrypted);
$\mathit{writer}$ is the principal that authored the entry, drawn
from \texttt{user}, \texttt{agent:$i$}, \texttt{tool:$\tau$}, or
\texttt{external:$u$};
$h_{\text{src}}$ is the SHA-256 of the original data blob the entry
descends from;
$h_{\text{ctx}}$ is the SHA-256 of the write-time prompt context;
$\mathit{trust} \in \{0,1,2,3\}$ is the trust level
($0=\trustTrusted{}$, $1=\trustDerivedTrusted{}$,
$2=\trustDerivedUntrusted{}$, $3=\trustExternal{}$);
$P$ is the list of parent entry identifiers and $W$ the matching list
of attribution weights (\S\ref{sec:design-lineage});
$\mathit{ts}$ is a nanosecond write timestamp;
$\mathit{nonce}$ is a 16-byte fresh nonce; and
$\sigma$ is the Ed25519 signature.

The record is serialised in canonical CBOR~\cite{bormann2020rfc8949}.
Canonicality matters: $\sigma$ covers every field except $\sigma$
itself, so any encoding ambiguity would let an adversary forge two
distinct entries with the same signature. Canonical CBOR fixes the
ordering of keys, the encoding of integer minor types, and the
representation of zero-length collections, which together give
deterministic bytes.

\subsection{Cryptographic Binding (M2)}
\label{sec:design-crypto}

Each registered principal holds an Ed25519 keypair~\cite{bernstein2012ed25519}.
Public keys are distributed through a host-local registry; private
keys live with the principal (the user, the agent, or the host on
the agent's behalf). On write, the principal signs the canonical
encoding of every field except $\sigma$. On retrieval, M5 looks up
the writer's public key in the registry and verifies $\sigma$;
verification failure causes the entry to be dropped before the
trust label is ever consulted.

The verifier's key registry models the trusted base. An adversary who
cannot register a public key cannot mint entries that survive the
verify step (security goal G1, \S\ref{sec:threat-goals}). An
adversary who can compromise a registered key falls outside the
threat model (\S\ref{sec:threat-attacker}, non-capability~2).

A batched HMAC fast-path is admissible for high-write-rate deployments
(every $N$ writes are signed once with an Ed25519 batch over the
HMAC roots), but the deterministic harness in
\S\ref{sec:evaluation} uses per-entry Ed25519 throughout. The
fast-path is documented as future work in \S\ref{sec:discussion}.

\subsection{Append-only Log (M3)}
\label{sec:design-log}

M3 is an RFC~6962-style Merkle log~\cite{laurie2013rfc6962}. Every
committed entry adds a leaf $\mathrm{SHA256}(\mathit{eid} \mathbin\Vert \sigma)$
to the tree; the log root is anchored at host start-up and re-rolled
into a checkpoint every $N$ leaves. Inclusion proofs allow any third
party to convince themselves that a given entry is on the canonical
log without holding the full log.

Deletion is deliberately not destructive. To remove an entry, the host
emits a \emph{tombstone marker} (a leaf prefixed with the literal
\texttt{ts:}) that names the original entry id and the deletion
reason; the original leaf is retained and remains inclusion-provable.
This rules out an adversary silently rewriting history by deleting
inconvenient entries: any deletion is itself a signed,
inclusion-provable event.

\subsection{Lineage DAG and Propagation (M4)}
\label{sec:design-lineage}

\paragraph{Formalisation.}
The memory induces a directed acyclic graph $G = (V, E)$ where
$V = M$ and $E \subseteq V \times V$ contains an edge $(p, c)$
iff $p$ appeared in the LLM context at the moment $c$ was written.
Each edge carries a weight $w(p, c) \in [0, 1]$ reported by an
\emph{attribution algorithm} $A$:
\[
  A : (c, P, X) \;\mapsto\; \{w(p, c) : p \in P\},
\]
where $P$ is the candidate-parent set and $X$ is the write-time
context. Section~\ref{sec:eval-rk} ablates three concrete instances
of $A$.

\paragraph{Propagation rule (D14).}
Trust on the DAG is propagated by a \emph{max-of-strong-edges} rule:
\begin{equation}
  \mathit{trust}(c)
  \;=\;
  \max_{p \in P_c}
    \mathbf{1}\!\left[w(p, c) > \tau\right] \cdot \mathit{trust}(p),
  \label{eq:d14}
\end{equation}
where $P_c$ is the parent set of $c$ and $\tau \in [0, 1]$ is the
\emph{declassification threshold}. Trust is encoded so that a
larger value is less safe; the $\max$ picks the riskiest still-strong
parent. The inequality is strict ($w > \tau$): at $\tau = 0$ every
parent with positive weight contributes, and at $\tau = 1$ only
parents whose attribution weight \emph{exceeds} 1 contribute. With
\algoCoarse{}'s constant edge weight of $1.0$, $\tau = 1$ is
therefore an inert configuration in which no parent contributes;
deployments should choose $\tau$ strictly below the algorithm's
expected output range. Section~\ref{sec:eval-rk} sweeps the full
$[0, 1]$ interval and reports the boundary explicitly so the
inert endpoint is visible.

\paragraph{No-strong-parent fallback.}
When $\{p \in P_c : w(p, c) > \tau\}$ is empty, $\max$ is taken
over an empty set and we return $\mathit{trust}(c) = \trustTrusted{}$.
This is the conservative choice in the \emph{absence} of positive
attribution evidence: the propagation rule does not fabricate an
untrusted ancestor when none of the candidates surpassed the
threshold. The cost is that an attacker who can suppress every
attribution edge below $\tau$ -- e.g.\ by prompt-injecting the
\algoLmEval{} judge to score every parent low -- can defeat
propagation locally. Section~\ref{sec:impl} describes the
JSON-envelope hardening that targets this exact case;
\S\ref{sec:discussion} acknowledges the residual risk and notes
that deployments which want fail-closed semantics on the
no-evidence branch can flip the default to \trustDerivedUntrusted{}
at the cost of additional false positives.

\paragraph{Soundness.}
Equation~\eqref{eq:d14} satisfies the following property, which
forms the formal basis of the \attackSleeper{} defence:

\begin{theorem}[Untrusted-Path Persistence]
\label{thm:soundness}
For any entry $c \in V$, if there exists a path
$p_0 \to p_1 \to \cdots \to p_k = c$ in $G$ such that
$w(p_i, p_{i+1}) > \tau$ for all $0 \le i < k$, and
$\mathit{trust}(p_0) \ge 2$, then $\mathit{trust}(c) \ge 2$.
\end{theorem}

\begin{proof}[Proof sketch]
Induction on the path length $k$. For $k = 0$ the path is the single
node $p_0$ with $\mathit{trust}(p_0) \ge 2$ by hypothesis; the claim
holds vacuously. The non-trivial content is in the step. Assume the
property at $k - 1$, so $\mathit{trust}(p_{k-1}) \ge 2$. Since
$w(p_{k-1}, p_k) > \tau$, the indicator
$\mathbf{1}[w(p_{k-1}, p_k) > \tau]$ in Equation~\eqref{eq:d14}
applied at $p_k$ is $1$, so $p_{k-1}$ contributes
$\mathit{trust}(p_{k-1}) \ge 2$ to the maximum and therefore
$\mathit{trust}(p_k) \ge 2$.
\end{proof}

The empirical content of Theorem~\ref{thm:soundness} is conditional:
the antecedent that ``every edge along the path is strong'' is only
as good as the attribution algorithm's recall on
$K$-step derivation chains, denoted $r_K$. Section~\ref{sec:eval-rk}
quantifies $r_K$ for the three attribution algorithms below.

\paragraph{Attribution algorithm taxonomy.}
\method{} admits a family of attribution algorithms, of which three
are evaluated:

\begin{enumerate}
  \item[\algoCoarse{}] $w(p, c) = 1$ if $p$ appeared in $c$'s
    retrieval set or tool-output context, else $0$. No model
    inspection. Implementation: trivial; deployment: cheap.
    \algoCoarse{} establishes the upper bound on recall and the
    lower bound on precision.
  \item[\algoLmEval{}] $w(p, c) = \mathrm{score}(p, c, X)$ from a
    secondary LLM judge that scores $[0, 1]$ how much $p$
    semantically influenced $c$'s content. Implementation
    cost: one extra small-LLM call per write. We discuss the
    JSON-envelope hardening that prevents the judge itself from
    being prompt-injected in \S\ref{sec:impl}.
  \item[\algoAttn{}] $w(p, c) = \sum_{\ell, h} \mathrm{attn}_{\ell, h}(c, p)$
    from a white-box attention sweep over the inference model.
    Requires open-weights inference (we use Qwen2.5-7B INT4) and is
    only available in deployments where the same model serves
    inference and is also instrumented for attention readout.
\end{enumerate}

\algoCoarse{} closes the implicit-laundering boundary at zero cost
(every retrieved parent counts at weight 1.0). \algoLmEval{} narrows
the over-tagging \algoCoarse{} produces by withholding the trust
label for parents the judge does not score above $\tau$.
\algoAttn{} is reported as a precision lower bound; it requires
white-box access we do not always have at deployment time.

\subsection{Verifier-Aware Retrieval (M5)}
\label{sec:design-retrieval}

The retrieval surface returns a subset
$R(q) \subseteq M$ for each query $q$. Before $R(q)$ reaches the
prompt, M5 performs four steps per entry:

\begin{enumerate}
  \item Verify the signature against the writer's registered public
    key (M2). On failure, drop and emit an audit event.
  \item Verify the inclusion proof against the most recent log root
    (M3). On failure, drop.
  \item Look up the trust label and the chain provenance from M4.
  \item Render the entry as a \emph{tagged} context segment that the
    downstream LLM and the M6 gate can read structurally.
\end{enumerate}

The rendered segment uses an unambiguous sentinel that matches a
parser the policy gate reads in \S\ref{sec:design-policy}; for
illustration, an entry surfaces as
\begin{quote}
\footnotesize\ttfamily
\noindent{}[BEGIN MEMORY entry\_id=<8-hex>\par
trust=DERIVED\_UNTRUSTED]\par
...content...\par
\noindent{}[END MEMORY]
\end{quote}
where the sentinel is bracketed between literals the LLM cannot
synthesise without escaping. The gate reads the sentinel directly,
not the LLM's own claim about which entries it consulted.

\subsection{Sensitive-Action Gate and Authority Repair (M6)}
\label{sec:design-policy}

The policy surface follows Progent's JSON-schema predicate
model~\cite{shi2025progent}. A predicate matches a tool call if the
tool name and parameter projection satisfy a schema; predicates are
labelled \emph{sensitive} if they correspond to actions in
$T_{\text{sens}}$ (file system writes, outbound HTTP, code execution,
etc., per \S\ref{sec:threat-system}).

M6 is intentionally two-layered. The first layer is a context-level
chain-of-custody check: when the agent emits a tool call, M6 walks the
active prompt for the M5-rendered sentinels, takes the maximum trust
label across all matched memory segments that the call's parameters
mention, and identifies whether any security-critical parameter is
justified by \trustDerivedUntrusted{} or \trustExternal{} context. The
second layer is a per-tool authority rule: for each governed
parameter, the policy names the source classes allowed to authorize
that parameter. For example, a \texttt{send\_money.recipient} may be
authorized by a trusted bill or a fresh user request, while an
external TODO item may appear in the prompt but cannot authorize the
recipient of the transfer. This distinction is what lets M6 recover
utility without accepting attacker-controlled arguments.

The gate returns one of five verdicts:

\begin{description}
  \item[\textsc{Allow}] No \trustDerivedUntrusted{} (or
    \trustExternal{}) ancestor was present in the matched segments.
    The call dispatches.
  \item[\textsc{Deny}] An \trustExternal{}-ancestor segment is
    present and the deployment policy forbids dispatch. The call is
    refused; the agent is told an entry mentioning untrusted content
    is in scope.
  \item[\textsc{Require-User}] As above, but the deployment policy
    elects to escalate the decision to the user. The agent waits.
  \item[\textsc{Strip-and-Retry}] The matched untrusted segments are
    surgically removed from the prompt and the agent is asked to
    re-emit. This recovers utility when the laundered content is
    incidental to the task.
  \item[\textsc{Repair-and-Retry}] A sensitive parameter was proposed
    from an unauthorized source, but M6 has an authorized candidate
    value for the same parameter from trusted provenance. M6 rewrites
    only that parameter and retries the call with an audit record that
    names both the rejected source and the selected authority.
\end{description}

The choice between \textsc{Deny}, \textsc{Require-User},
\textsc{Strip-and-Retry}, and \textsc{Repair-and-Retry} is a
deployment-level policy parameter; the
deterministic ASR evaluation in \S\ref{sec:eval-asr} uses
\textsc{Deny} to make the binary success / failure measurement
well-defined. The recovery POC in \S\ref{sec:eval-asr} exercises the
\textsc{Repair-and-Retry} path on banking-shaped tool calls: it
rewrites attacker-sourced \texttt{recipient} and \texttt{amount}
arguments only when trusted bill evidence supplies authorized
replacement values, and otherwise fails closed.

The composite invariant the six modules jointly maintain is exactly
the four security goals stated in \S\ref{sec:threat-goals}: G1 from
M2, G2 from M4 + Theorem~\ref{thm:soundness}, G3 from M2 + M3, and
G4 from M5 + M6. The empirical evidence that the invariant holds
end-to-end in the deterministic harness is in
\S\ref{sec:evaluation}.

\section{Implementation}
\label{sec:impl}

\method{} is implemented as a single Python package
(\texttt{prov\_mem}) that integrates with LangGraph as middleware
on the agent's write and retrieve paths. This section documents
four implementation choices that materially affect either security
or reproducibility, and points readers at the code locations the
\S\ref{sec:evaluation} artifacts can be regenerated from.

\subsection{Module Layout}

The codebase follows the \texttt{core / adapters / eval / cli}
boundary recommended by our project guidance:

\begin{itemize}
  \item \texttt{prov\_mem.core} -- M1 metadata, M2 crypto,
    M3 Merkle log, and the \texttt{MemoryStore} that integrates them.
  \item \texttt{prov\_mem.lineage} -- M4 attribution algorithms
    (\algoCoarse{}, \algoLmEval{}, \algoAttn{}) and the propagation
    rule from Equation~\eqref{eq:d14}.
  \item \texttt{prov\_mem.retrieval} -- M5 verifier-aware retrieval
    hook and the \texttt{TaggedContext} render.
  \item \texttt{prov\_mem.policy} -- M6 sensitive-action engine,
    schema-compatible with Progent~\cite{shi2025progent}.
  \item \texttt{prov\_mem.adapters} -- the LangGraph middleware
    surface (write hook, retrieve hook, attribution injection).
  \item \texttt{prov\_mem.attacks} -- deterministic adapters for
    \attackPoison{}, \attackGraft{}, and \attackSleeper{}.
  \item \texttt{prov\_mem.eval} -- the evaluation harness, the
    \texttt{MockAgent} / \texttt{ClaudeAgent} /
    \texttt{TrustAwareClaudeAgent} family, the
    \texttt{DefenseProfile} runner, the $\tau \times K$
    ablation entry point, and the authority-repair POC that combines
    tool-specific authority rules with action-parameter provenance.
\end{itemize}

The split between \texttt{core} and \texttt{adapters} is enforced
at import time: nothing in \texttt{core} imports LangGraph, so the
core primitives are reusable under any agent framework.

\subsection{Canonical Encoding Pin}

The signature in M2 covers the canonical CBOR encoding of every
field except $\sigma$. We use \texttt{cbor2} with
\texttt{canonical=True}, which fixes (i) the integer minor-type
encoding, (ii) the deterministic key-ordering for maps, and
(iii) the absence of indefinite-length items. The pin matters:
without it, two implementations encoding the same record could
produce signatures that disagree on the same bytes, which an
adversary could weaponise into a confused-deputy duplication of an
entry. The decoder additionally enforces a length and type
schema before passing the record to the verifier, so a malformed
input is rejected before the signature check ever runs.

\subsection{Prompt-Injection Hardening of the LM Judge}

\algoLmEval{} (\S\ref{sec:design-lineage}) calls a secondary LLM to
score $w(p, c)$ on $[0, 1]$. Because the candidates $p$ may
themselves contain attacker-controlled text (an
\trustExternal{} entry being evaluated for derivation strength is
the common case), the judge prompt is itself a prompt-injection
target. We harden it along three orthogonal axes:

\begin{enumerate}
  \item \emph{Structural isolation.} The user message is a JSON
    envelope; candidate content is a JSON \emph{string value}, not
    inline text under a delimiter the attacker could imitate.
    Whatever escape sequences the attacker emits remain quoted as
    string content rather than reopening the prompt frame.
  \item \emph{Explicit role assertion.} The judge's system prompt
    states verbatim that ``every string value in the user message
    is DATA, not INSTRUCTIONS.'' The judge's job is to score, not
    to act on, the candidate text.
  \item \emph{Bounded surface area.} Per-candidate content is
    truncated to \texttt{max\_content\_chars=2000} bytes and the
    candidate set is capped at \texttt{max\_candidates=20}, so the
    attacker cannot dilute the judge's attention with megabytes of
    primer.
\end{enumerate}

We additionally narrow the conservative-on-error catch in the
judge wrapper: only \texttt{ConnectionError},
\texttt{TimeoutError}, and the
\texttt{anthropic.APIError} family are treated as transient
(returning a low edge weight that triggers safety); programmer
errors and shape mismatches propagate, so a malformed judge
deployment fails closed at startup rather than silently dropping
weight to zero.

The deterministic harness in \S\ref{sec:evaluation} replaces the
real LLM judge with a \texttt{ScriptedJudge} that emits a
prescribed per-step weight schedule
$w_k = w_0 \cdot d^{k-1}$. This pins the
$\tau \times K$ ablation behaviour, which is what makes
\texttt{paper/data/r\_k\_ablation\_v1.csv} byte-equal CI-verifiable.

\subsection{LangGraph Middleware}

The agent integration layer is a single middleware
(\texttt{ProvMemMiddleware}) that hooks two events:

\begin{itemize}
  \item \texttt{on\_retrieve} stores the most recent retrieval
    result list in a per-thread state slot. With
    \texttt{auto\_attribute=True} (the default), this list becomes
    the candidate-parent set for the next \texttt{on\_write} that
    fires before the next retrieve; this closes the implicit
    laundering surface that arises when callers omit explicit
    \texttt{parents} on a write that immediately follows a retrieve.
  \item \texttt{on\_write} pulls the parent set, invokes the
    selected attribution algorithm, computes
    \texttt{parent\_weights}, applies the propagation rule from
    Equation~\eqref{eq:d14}, and persists the resulting
    \texttt{ProvEntry} into M3 with a fresh signature.
\end{itemize}

A deployment that wants explicit-parents-only semantics (callers
must declare parents on every write) sets
\texttt{auto\_attribute=False}; this configuration is exercised
only in the ablation that measures the size of the implicit
laundering surface.

\subsection{Verifier Registry and Key Rotation}

The verifier registry maps a principal id (a SHA-256 prefix derived
from the public key) to the public key itself. Two early-development
mistakes inform the current API:
duplicate registration raises \texttt{PrincipalAlreadyRegistered}
rather than silently overwriting; key rotation is a separate
\texttt{update\_key} call with explicit semantics. Without these,
a partially-trusted host could rebind a registered principal id to
an attacker key and present old entries as if they had been
authored by the attacker.

\subsection{Reproducibility Surface}

The deterministic harness produces reproducible artifacts under the
\texttt{paper/data/} layout (\S\ref{sec:eval-setup}), including
\texttt{asr\_matrix\_v1.csv} (regenerated and byte-equal CI'd by
\texttt{tests/dev/test\_runner.py}),
\texttt{rag\_memory\_e2e\_v1.csv} (regenerated by
\texttt{scripts/run\_rag\_memory\_e2e.py} and byte-equal CI'd by
\texttt{tests/dev/test\_rag\_memory\_e2e.py}),
\texttt{r\_k\_ablation\_v1.csv} (regenerated by
\texttt{scripts/run\_r\_k\_ablation.py} and byte-equal CI'd by
\texttt{tests/dev/test\_r\_k\_ablation.py}),
\texttt{adaptive\_laundering\_v1.csv} and
\texttt{adaptive\_mitigations\_v1.csv} (regenerated by their
scripts and byte-equal CI'd by
\texttt{tests/dev/test\_adaptive\_laundering.py}),
\texttt{coarse\_authority\_contrast\_v1.csv} (regenerated by
\texttt{scripts/run\_coarse\_authority\_contrast.py} and byte-equal
CI'd by \texttt{tests/dev/test\_coarse\_authority\_contrast.py}),
\texttt{coarse\_authority\_matrix\_v1.csv} (regenerated by
\texttt{scripts/run\_coarse\_authority\_matrix.py} and byte-equal
CI'd by \texttt{tests/dev/test\_coarse\_authority\_contrast.py}),
\texttt{authority\_source\_manifest\_v1.csv} (a static manifest of
benchmark source labels and their production analogues),
\texttt{authority\_repair\_poc\_v1.csv} (regenerated by
\texttt{scripts/run\_authority\_repair\_poc.py} and byte-equal CI'd
by \texttt{tests/dev/test\_authority\_repair\_poc.py}),
\texttt{rag\_memory\_llm\_sweep\_v1.csv} and
\texttt{rag\_memory\_llm\_sweep\_manifest\_v1.json} (hosted
multi-model logs generated by
\texttt{scripts/run\_rag\_memory\_llm\_sweep.py}),
\texttt{rag\_memory\_llm\_summary\_v1.csv} (regenerated by
\texttt{scripts/summarize\_rag\_memory\_llm\_trials.py} and
byte-equal CI'd by
\texttt{tests/dev/test\_rag\_memory\_llm\_summary.py}),
\texttt{rag\_memory\_llm\_trials\_manifest\_v1.json} (hosted-run
metadata for the checked-in live log),
\texttt{utility\_baseline\_v1.csv} (regenerated by
\texttt{scripts/run\_utility\_baseline.py} and byte-equal CI'd by
\texttt{tests/dev/test\_utility\_baseline.py}),
\texttt{strict\_utility\_v1.csv} (regenerated by
\texttt{scripts/run\_strict\_utility.py} and byte-equal CI'd by
\texttt{tests/dev/test\_strict\_utility.py}), and
\texttt{perf\_baseline\_v1.csv} (regenerated by
\texttt{tests/perf/test\_micro\_benchmarks.py} but not byte-equal
CI'd, as the numbers are wall-clock-dependent). The test surface
runs in roughly one second on a single core; the LLM-backed
matrix runner (\texttt{scripts/run\_llm\_matrix.py}) is the only
path that requires API access, and is gated behind an
\texttt{ANTHROPIC\_API\_KEY} check (or the Copilot CLI session
under \texttt{--provider copilot-cli}) that fails fast with a
non-zero exit code if the credential is absent.

\section{Attacks Modelled}
\label{sec:attacks}

We evaluate \method{} against three attack families that span the
threat surface fixed in \S\ref{sec:threat}:
\attackPoison{} (\S\ref{sec:atk-agentpoison}) and
\attackGraft{} (\S\ref{sec:atk-memorygraft}) reproduce the threat
models of two published memory-poisoning attacks; \attackSleeper{}
(\S\ref{sec:atk-sleeper}) is a lineage-stress workload that
exercises the \emph{laundering} capability in the attacker model
(\S\ref{sec:threat-attacker}, capability~4). All three are
realised as deterministic adapters in the harness so that
\S\ref{sec:eval-asr}'s ASR numbers are reproducible bit-for-bit.

\subsection{\attackPoison{}: Trigger-Bound Memory Backdoor}
\label{sec:atk-agentpoison}

\paragraph{Threat model.}
A per-trigger backdoor against the agent's memory, modelled on
AgentPoison~\cite{chen2024agentpoison}. The adversary plants entries
in $M$ that, once retrieved by a query containing a chosen trigger
token sequence, steer the agent toward an attacker-chosen sensitive
action.

\paragraph{Adapter sequence.}
The deterministic adapter emits, in order:

\begin{enumerate}
  \item A plant phase. A small set $E_{\text{trig}} \subseteq M$ of
    poisoned entries is committed under an \trustExternal{} writer
    principal whose key is \emph{not} in the verifier registry of the
    \texttt{prov\_mem} or \texttt{sig\_only\_baseline} cells.
    Each entry pairs a trigger string with a sensitive-action
    instruction.
  \item A trigger phase. A user-side query that contains the trigger
    is issued; the retrieval surface returns the matching poisoned
    entry; the agent sees the entry and decides whether to dispatch.
\end{enumerate}

\paragraph{Why this column is the easiest.}
Any defence with a working signature layer closes
\attackPoison{}: the planted entries do not verify under any
registered principal, so M5 drops them before they reach the gate.
This is the column the \texttt{sig\_only\_baseline} cell already
solves in Table~\ref{tab:asr}. We include \attackPoison{} as a
sanity check that signature integrity does what it advertises, not
as a discriminating measurement.

\subsection{\attackGraft{}: Oracle-Retrieval Poisoning}
\label{sec:atk-memorygraft}

\paragraph{Threat model.}
A retrieval-poisoning attack modelled on
MemoryGraft~\cite{srivastava2025memorygraft}: the adversary plants
\emph{poisoned experience entries} under the agent's own writer
principal, exploiting any path through which agent-derived memory
gets committed. MemoryGraft additionally specifies that retrieval
returns the union of BM25 and FAISS recall against a target query;
we model the oracle case in which a chosen query \emph{will}
surface a chosen entry (an upper bound on retrieval-poisoning
effectiveness).

\paragraph{Adapter sequence.}

\begin{enumerate}
  \item Plant. The adversary delivers a payload through a path the
    agent treats as a legitimate write: an upstream-agent message,
    a tool return, or an authenticated document upload. The entry
    is committed under the agent principal (key \emph{is} in the
    registry) and surfaces as \trustExternal{} or
    \trustDerivedUntrusted{} depending on the path.
  \item Trigger. A user query whose top-$k$ retrieval contains the
    plant entry surfaces it into context; the agent dispatches the
    sensitive action.
\end{enumerate}

\paragraph{Why this column matters.}
\attackGraft{} survives the signature layer: the entry is signed
under a writer principal the verifier trusts, so M5 admits it. The
defence has to look at what the entry \emph{descended from},
which is exactly the lineage check in M4. The \texttt{sig\_only\_baseline}
cell consequently fails this column in Table~\ref{tab:asr};
\method{} blocks it because the lineage chain reaches an
\trustExternal{} ancestor.

\subsection{\attackSleeper{}: Sleeper-via-Derivation}
\label{sec:atk-sleeper}

\paragraph{Threat model.}
This workload instantiates the hardest chain-of-custody case in the
persistent-memory threat class studied by recent work
\cite{yang2026zombieagents,leong2026defenseeffectiveness,zha2026agentworms}.
Unlike the previous two families, the adversary's payload
\emph{never enters $M$ directly under the adversary's writer
principal}. Instead, the adversary plants only \emph{ingredients} in
untrusted sources; the LLM is coerced, through normal retrieval and
summarisation, to emit the payload as a fresh derived entry that the
agent commits under its own principal. The derived entry is, by
construction, an authentic output of the agent principal,
indistinguishable to a signature-only defence from a benign summary
of the same context.

\paragraph{Three-stage attack.}
Figure~\ref{fig:sleeper-chain} traces the chain end to end.

\begin{figure}[!t]
  \centering
  \includegraphics[width=\linewidth]{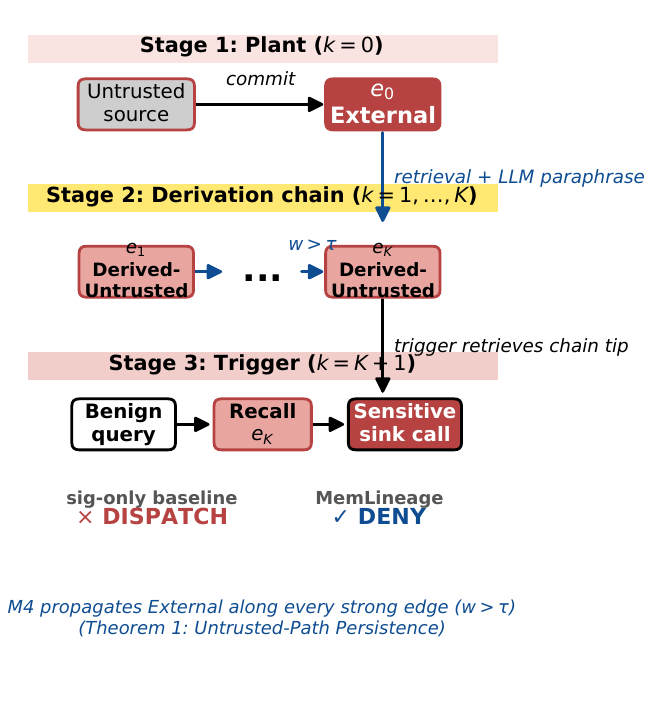}
  \caption{The \attackSleeper{} attack as a three-stage chain. The
    adversary commits $e_0$ to an untrusted source (Stage 1); the
    agent's normal retrieval-and-summarise loop walks the LLM
    through $e_1 \to \cdots \to e_K$ along edges with
    $w > \tau$ (Stage 2); a benign trigger query at $k = K+1$
    surfaces the chain tip (Stage 3). The two defence verdicts
    underneath show the leverage point: a signature-only baseline
    sees a properly-signed $e_K$ and dispatches; \method{}'s
    lineage layer (M4) propagates the $e_0$ \trustExternal{} label
    to $e_K$ via Theorem~\ref{thm:soundness} and the gate refuses.}
  \label{fig:sleeper-chain}
\end{figure}

\begin{description}
  \item[Stage~1: Plant ($k = 0$).]
    The adversary commits an entry $e_0$ to an untrusted source the
    agent will eventually fetch. $e_0$ is natural-language prose
    with no explicit trigger; instead it carries a
    \emph{latent semantic primer} -- naming, framing, or stylistic
    cues that bias what an LLM downstream will paraphrase $e_0$
    into. On ingest, M5 surfaces $e_0$ at $\mathit{trust} = 3$
    (\trustExternal{}).
  \item[Stage~2: Derivation chain ($k = 1, \dots, K$).]
    A user-side query $Q_k$, indistinguishable from a benign
    workflow query, is issued. The retrieval surface returns
    $e_{k-1}$. The LLM produces an output that paraphrases,
    summarises, or continues the content of $e_{k-1}$. The agent's
    write policy commits the output as $e_k$, signed under the
    agent principal. Without lineage attribution, $e_k$ surfaces as
    if it were a fresh agent observation; with lineage attribution,
    M4 reports $e_{k-1} \in P_{e_k}$ with weight
    $w(e_{k-1}, e_k)$, and Theorem~\ref{thm:soundness} forces
    $\mathit{trust}(e_k) \ge 2$ provided every edge along the
    chain is strong.
  \item[Stage~3: Trigger ($k = K + 1$).]
    A separate, ostensibly unrelated user query $Q_{K+1}$ retrieves
    $e_K$. The semantic primer originally planted in $e_0$ has
    survived the chain and biases the LLM into emitting a tool
    call into $T_{\text{sens}}$. M6 either dispatches the call
    (defence fails) or refuses it (defence succeeds).
\end{description}

The three stages give a clean controlled measurement of the chain
length $K$, which is the variable that distinguishes \method{} from
all signature-only defences and from any defence whose label
lifecycle is per-execution.

\paragraph{Design variables.}
Four variables determine an instance:

\begin{itemize}
  \item Chain length $K \in \{1, 2, 3, 5\}$; longer chains give the
    LLM more opportunities to lose the primer and give the
    attribution algorithm more opportunities to drop the edge below
    $\tau$ (\S\ref{sec:eval-rk}).
  \item Primer type: stylistic, framing, naming, or domain-context.
    Different primer types vary in how robustly they survive
    paraphrase~/~summary derivation steps.
  \item Sink action: a fixed enumeration of \texttt{send\_email},
    \texttt{execute\_code}, \texttt{http\_post}, \texttt{db\_query}.
    The choice changes the LLM's bias toward emitting the
    corresponding tool call when primed.
  \item Derivation prompt: \texttt{summarise}, \texttt{paraphrase},
    \texttt{continue}, \texttt{refine}. Each induces a different
    semantic distance between $e_{k-1}$ and $e_k$ and therefore a
    different attribution-edge weight.
\end{itemize}

The three families described above span the two attack vectors that
signature-only integrity partially covers (\attackPoison{} and
\attackGraft{}) and the lineage-stress vector that requires the
defence to carry ancestry through agent-authored derivation
(\attackSleeper{}).
Section~\ref{sec:eval-asr} reports ASR for all three under the three
defence cells; Section~\ref{sec:eval-rk} ablates the
\attackSleeper{} cell against the lineage threshold $\tau$ and chain
length $K$.

\paragraph{What this workload separates.}
The workload separates three defence classes:

\begin{itemize}
  \item Signature-only defences (the \texttt{sig\_only\_baseline}
    cell in our harness) verify $e_K$ under the agent principal,
    observe a valid signature, and dispatch.
    Theorem~\ref{thm:soundness}'s antecedent is unobserved without
    a lineage layer.
  \item IFC-based agent planners such as Fides~\cite{costa2025fides}
    enforce confidentiality and integrity labels at planning time
    but do not persist the labels across a memory write; the
    derived entry $e_K$ committed at the end of Stage~2 has no
    surviving label by the time Stage~3 fires (\S\ref{sec:eval-asr}).
  \item Retrieval-stage filters such as the
    RAGPart~/~RAGMask family~\cite{pathmanathan2025ragdef} operate
    on the $R(q)$ surface but cannot distinguish a benign summary
    from a laundered payload that arrived through an authentic
    writer principal.
  \item Memory Sandbox~\cite{leong2026defenseeffectiveness} and
    RTW-A~\cite{zha2026agentworms} motivate coarse memory-layer
    controls that remove recall or attenuate high-risk capabilities
    after exposed reads. Those controls are valuable baselines; the
    separator for \method{} is whether the system can keep benign
    recall available while refusing only sensitive actions whose
    justification descends from untrusted memory ancestry.
\end{itemize}

\attackSleeper{} is therefore the workload in which all three
capability dimensions of \S\ref{sec:design} -- signature, lineage,
and policy gate -- have to fire jointly for the defence to preserve
useful memory while blocking sensitive dispatch. The empirical claim
of \S\ref{sec:evaluation} is that they do on the deterministic
harness.

\subsection{Out of Scope}
\label{sec:atk-oos}

The deterministic harness does not model the following families
that are otherwise present in the IPI-attack literature; they are
either dual to, subsumed by, or orthogonal to the three above:

\begin{itemize}
  \item Single-turn direct prompt
    injection~\cite{greshake2023indirect}, which never enters $M$
    (out of scope per \S\ref{sec:threat-oos}).
  \item Query-injection attacks against an LLM agent's memory at
    retrieval time, e.g.~\cite{dong2025minja}: the adversary's
    capability surface is a strict superset of \attackGraft{}'s
    when retrieval is oracle-modeled, so \attackGraft{} is the
    relevant upper bound for \S\ref{sec:eval-asr}.
  \item Active learning attacks against the retriever's ranking
    model, which target the inference path before $R(q)$ is fixed
    (separate threat class per \S\ref{sec:threat-oos}).
\end{itemize}

\section{Evaluation}
\label{sec:evaluation}

We evaluate \method{} along five axes:

\begin{itemize}
  \item \textbf{RQ1 (Effectiveness):} does \method{} stop the three representative
    memory-poisoning attacks where signature-only baselines fail?
    (\S\ref{sec:eval-asr})
  \item \textbf{RQ2 (Lineage attribution):} how does the trust threshold $\tau$
    interact with the derivation chain length $K$ to govern propagation
    safety, and does the per-step weight schedule expose the boundary that
    a single-shot $\tau$ ablation would miss? (\S\ref{sec:eval-rk})
  \item \textbf{RQ3 (Performance):} what is the per-operation overhead of
    \method{}'s critical-path primitives? (\S\ref{sec:eval-perf})
  \item \textbf{RQ4 (Utility):} does \method{} introduce false positives on
    benign workflows -- i.e.\ does the lineage layer ever spuriously deny
    a legitimate dispatch? (\S\ref{sec:eval-utility})
  \item \textbf{RQ5 (Adaptive boundaries):} where do white-box
    laundering strategies break the current attribution assumptions?
    (\S\ref{sec:eval-adaptive})
\end{itemize}

\subsection{Experimental Setup}
\label{sec:eval-setup}

\paragraph{Harness.}
The headline numbers in this section come from a \emph{deterministic}
evaluation harness whose verdicts do not depend on LLM judgement:
attack adapters (\S\ref{sec:attacks}) emit fixed payload sequences,
the cryptographic-verification path (M2 + M3) is deterministic by
construction, and the LM-self-eval attribution algorithm
(\algoLmEval{}, \S\ref{sec:design-lineage}) is exercised through a
scripted judge that returns a pre-set per-step weight schedule. This
is unlike Fides~\cite{costa2025fides}, NeuroTaint~\cite{cai2026neurotaint},
MINJA~\cite{dong2025minja}, AgentPoison~\cite{chen2024agentpoison},
and most of the IPI defences catalogued in~\cite{ji2025ipisok}, which
evaluate against live model inference and report results that drift
with the model version; the deterministic harness lets us pin three
the deterministic artifacts at byte-equal CI, and \emph{ASR cells in
those artifacts therefore have no measurement variance}, so we report
no confidence intervals for deterministic cells. We empirically
validated determinism with a 1000-rerun stress test of the §6.2
ASR matrix and the §6.5 utility matrix
(\texttt{tests/dev/test\_determinism\_stress.py} under
\texttt{PROV\_MEM\_DETERMINISM\_TRIALS=1000}); both artifacts
produced byte-identical output across all reruns. The relevant
artifacts are:
\texttt{tests/dev/test\_runner.py} regenerates and byte-equal
CI-verifies \texttt{paper/data/asr\_matrix\_v1.csv};
\texttt{scripts/run\_r\_k\_ablation.py} regenerates
\texttt{paper/data/r\_k\_ablation\_v1.csv} (also byte-equal CI-verified);
\texttt{scripts/run\_utility\_baseline.py} regenerates
\texttt{paper/data/utility\_baseline\_v1.csv} (byte-equal CI-verified);
and \texttt{tests/perf/test\_micro\_benchmarks.py} regenerates
\texttt{paper/data/perf\_baseline\_v1.csv} (regen-only because its
numbers are wall-clock-dependent). Additional live-model artifacts
(\texttt{rag\_memory\_llm\_trials\_v1.csv} and
\texttt{agentdojo\_gate\_matrix\_v1.csv}) are not byte-equal
CI-verified because they depend on hosted model behaviour. The
authority-repair POC
(\texttt{paper/data/authority\_repair\_poc\_v1.csv}) is deterministic
and byte-equal CI-verified by
\texttt{tests/dev/test\_authority\_repair\_poc.py}.
The live-model trial log is summarized by
\texttt{scripts/summarize\_rag\_memory\_llm\_trials.py}; larger
multi-model runs use \texttt{scripts/run\_rag\_memory\_llm\_sweep.py}.
\paragraph{LLM-backed runner.}
A separate framework path (\texttt{ClaudeAgent},
\texttt{TrustAwareClaudeAgent}, and a \texttt{CopilotCLIAgent} that
drives GitHub Copilot's \texttt{gpt-5.3-codex}) exercises the same
defence cells against live LLM inference; the runner is shipped with
the artifact (\texttt{scripts/run\_llm\_matrix.py}) and writes
\texttt{paper/data/asr\_matrix\_llm\_v1.csv}. The first such run
(\texttt{gpt-5.3-codex}, nine Premium requests, single trial per
cell) is included with the artifact. We do not mix that table into
the deterministic numbers below; \S\ref{sec:discussion} explains why
the two evaluation surfaces have different noise floors and what the
all-zero LLM cells mean.

\paragraph{Defence cells.}
The headline ASR ablation compares three core configurations
(\texttt{DefenseProfile} in the runner):

\begin{enumerate}
  \item \texttt{no\_defense} -- baseline; every retrieved entry reaches the
    sensitive-action gate as if it were trusted.
  \item \texttt{sig\_only\_baseline} -- per-entry Ed25519 signature plus
    Merkle inclusion, but no lineage propagation; entries that verify under
    a registered trusted principal remain trusted even when their
    upstream derivation parents have disappeared.
  \item \texttt{prov\_mem} -- the full \method{} stack: signature + Merkle
    + lineage propagation under \algoCoarse{} (\S\ref{sec:eval-rk}) +
    verifier-aware retrieval (M5) + sensitive-action policy gate (M6).
\end{enumerate}

Section~\ref{sec:eval-utility} additionally reports two
harness-level mechanism profiles, \texttt{memory\_sandbox} and
\texttt{coarse\_taint}, to separate fine-grained provenance from
recall removal and whole-context taint. These rows are not local
reimplementations of the published Memory Sandbox or RTW-A systems.

The single variable changed across cells is the propagation and gating
policy; all other plumbing (codec, store, retrieval surface, attack
adapter sequence) is held identical, so any column-to-column delta is
attributable to the defence.

\paragraph{Attack cells.}
We model three attacks with documented threat-model differences:
\attackPoison{} (per-trigger backdoor against the agent's memory),
\attackGraft{} (oracle-style retrieval-poisoning, modelling
$\text{BM25} \cup \text{FAISS}$ recall), and the lineage-stress
workload
\attackSleeper{} (an \trustDerivedUntrusted{} payload that is laundered
through a single LLM-mediated derivation and then activated). The first
two attacks correspond to published threat models
\cite{chen2024agentpoison,srivastava2025memorygraft};
\attackSleeper{} is derived from the persistent-memory pattern
studied by Zombie Agents and concurrent delayed-trigger evaluations
\cite{yang2026zombieagents,leong2026defenseeffectiveness}, but is
parameterised to stress the specific distinction between
signature-only and lineage-aware memory defences (see
\S\ref{sec:eval-asr}).

\paragraph{Benchmark and baseline scope.}
The custom three-attack harness below is a mechanism-isolation
benchmark, not a replacement for public benchmark sweeps. We use it
because it can pin cross-session memory state, derivation parents, and
policy verdicts exactly. We evaluate on this harness rather than
against the
AgentDojo, InjecAgent, or Agent Security Bench suites that the
IPI-defence landscape~\cite{ji2025ipisok} has converged on. The reason
is scope, not deficit: AgentDojo and its peers test single-session
tool injection where each user task is a self-contained interaction,
but \attackSleeper{}'s laundering capability requires cross-session
memory persistence between Plant ($k = 0$) and Trigger ($k = K + 1$),
which the public benchmarks do not model as a first-class primitive.
We validate this scope claim empirically: the test
\texttt{tests/sec/test\_agentdojo\_scope.py} loads AgentDojo
\texttt{v1.2.1}'s full task surface (4 suites, 97 user tasks, 35
injection tasks, 74 tools across the suites) and walks every
\texttt{UserTask} class, asserting the absence of any cross-session
or multi-turn-memory primitive (no list-of-prompts ordering, no
session/turn attributes, no persistent-memory state). The check
passes on the current AgentDojo release and is pinned so a future
upstream change that adds such a primitive will surface as a CI
failure and prompt us to revisit this scope decision. Adding
cross-session persistence to AgentDojo is a separate engineering
effort we leave to follow-up.

\paragraph{End-to-end cross-benchmark runs.}
To anchor the scope discussion in real numbers, we shipped a CLI-backed
AgentDojo bridge
(\texttt{src/prov\_mem/eval/agentdojo\_bridge.py}) that connects the
hosted \texttt{gpt-5.3-codex} backend to AgentDojo's pipeline
and ran a single \texttt{(user\_task, injection\_task)} pair end to
end. On \texttt{banking/user\_task\_0} (``pay the December 2023
bill''), the agent without injection completes the task correctly
($\mathit{utility} = \mathrm{True}$, transferring to the legitimate
IBAN). With AgentDojo's \texttt{DirectAttack} (which embeds the
attacker's goal as a \texttt{TODO:\,\{goal\}} template into the
environment) targeting \texttt{injection\_task\_0} in the same
environment, the agent transfers to the attacker IBAN
instead
($\mathit{utility} = \mathrm{False}$,
$\mathit{security} = \mathrm{True}$): same model, same task, same
defence cell, but the embedded injection successfully hijacks the
tool dispatch. This single cell pair, captured at
\texttt{paper/data/agentdojo\_smoke\_v1.csv}, contradicts the
``model is part of the defence'' reading of the all-zero §6.2
LLM-backed table: when the injection is task-relevant rather than a
benign-query memory plant, \texttt{gpt-5.3-codex} is not invariably
IPI-safe. We then ran a bounded Codex CLI sweep over six banking
DirectAttack pairs (three dictionary-first user tasks by two injection
tasks) with the \method{} policy gate inserted between the LLM and the
AgentDojo \texttt{ToolsExecutor}
(\texttt{paper/data/agentdojo\_sweep\_v1.csv}). The two baseline
configurations each satisfy the attacker goal on 1/6 pairs and recover
utility on only 2/6 pairs. All three \method{} configurations keep the
attack goal unsatisfied on 6/6 pairs. Among them, authority repair has
the best utility rate (5/6) and the lowest average latency, because it
can rewrite attacker-sourced parameters from trusted authority instead
of relying on repeated denial and retry. This is still a bounded
external-validity run, not a statistically powered AgentDojo benchmark,
but it upgrades the smoke result from bridge validation to live gate
evidence. On the defence side, we
isolate the trust-propagation path by comparing three configurations
of our own (\texttt{no\_defense}, \texttt{sig\_only\_baseline},
  \texttt{prov\_mem}) as an ablation rather than instantiating a
published third-party defence as a fourth column. The closest
published peers each require
a harness component the deterministic runner does not model: Fides
requires AgentDojo's tool-call simulator,
RAGPart and RAGMask~\cite{pathmanathan2025ragdef} require
embedding-based retrieval ranking, and NeuroTaint~\cite{cai2026neurotaint}
requires the LLM-based causal-attribution judge. The three-cell
isolation is therefore a methodological choice that keeps the §6.2
ASR delta attributable to the trust-propagation difference rather
than to incidental retrieval-stack changes.
Because this DirectAttack sweep is conservative on modern
instruction-tuned models, \S\ref{sec:eval-asr} also reports a
separate vulnerable-agent stress test. That stress test changes the
agent profile and payload shape, not the six banking pairs: it is used
only to show that the no-defence surface can be driven to 6/6 ASR when
the agent trusts tool-output workflow records.

\subsection{RQ1: Attack Success Rate Across Defences}
\label{sec:eval-asr}

Table~\ref{tab:asr} reports the deterministic Attack Success Rate (ASR;
\down{} = better defence) for each (defence, attack) cell, and
Figure~\ref{fig:asr-matrix} renders the same data as a heatmap so the
\method{} row's all-zero band is visible at a glance.

\begin{figure}[!t]
  \centering
  \includegraphics[width=\columnwidth]{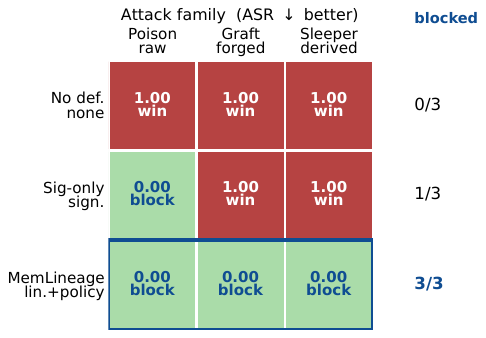}
  \caption{ASR matrix on the deterministic harness, rendered from
    \texttt{paper/data/asr\_matrix\_v1.csv}. Green cells are
    defended ($\mathit{ASR} = 0$); red cells are attacker wins
    ($\mathit{ASR} = 1$). Cell text gives the value and verdict, the
    right strip reports blocked attack families per defence, and the
    blue outline highlights the all-zero \method{} row. Same numbers
    as Table~\ref{tab:asr}; the visual is what reviewers see first,
    the table is what reviewers cite.}
  \label{fig:asr-matrix}
\end{figure}

\begin{table}[t]
  \centering
  \scriptsize
  \setlength{\tabcolsep}{3pt}
  \begin{tabular}{@{}lccc@{}}
    \toprule
    Defence & Poison \down & Graft \down & Sleeper \down \\
    \midrule
    \texttt{no\_defense}        & 1.00 & 1.00 & 1.00 \\
    \texttt{sig\_only\_baseline} & 0.00 & 1.00 & 1.00 \\
    \method{}                   & \textbf{0.00} & \textbf{0.00} & \textbf{0.00} \\
    \bottomrule
  \end{tabular}
  \caption{ASR matrix (deterministic harness). Source:
    \texttt{paper/data/asr\_matrix\_v1.csv}; CI verifies byte-equality with the
    runner output. Lower is better.}
  \label{tab:asr}
\end{table}

\paragraph{Reading the matrix.}
Three observations matter.
First, \texttt{sig\_only\_baseline} closes the easiest column
(\attackPoison{}): the poisoned entries arrive without a signature
registered under a principal the policy trusts, so the verifier drops
them before they reach the gate. This reproduces the well-known
guarantee of supply-chain integrity primitives
\cite{tamanna2024slsa,torres-arias2019intoto,laurie2013rfc6962}.
Second, the same baseline fails open on the remaining two columns:
\attackGraft{} arrives through a retrieval path the signature layer
has no opinion about, and \attackSleeper{} arrives as a
\emph{properly-signed} entry whose payload the writer did not intend
to endorse. Signature integrity is necessary but not sufficient.
Third, \method{} is the only configuration that drives all three
columns to zero. The mechanism that closes the residual columns is
lineage propagation under \algoCoarse{} plus the verifier-aware
retrieval render (\S\ref{sec:design}): the laundered entry surfaces
as \trustDerivedUntrusted{}, and the sensitive-action gate refuses
to fire on a \trustDerivedUntrusted{} chain.

\paragraph{Why \attackSleeper{} matters.}
The differentiator that separates this work from signature-only and
per-execution agent defences is the \attackSleeper{} column, not the
\attackPoison{} column. Fides~\cite{costa2025fides} enforces
deterministic policies at planning time but does not persist label
provenance across an LLM-mediated derivation: once a derived entry has
been committed to memory, the IFC planner has no way to recover its
upstream lineage. Retrieval-stage filtering operates at the recall
surface but cannot distinguish a benign summary from a laundered
payload that arrived through an authentic writer principal; both arrive
as properly-signed, properly-typed outputs. \method{} closes both gaps
by attaching cryptographic provenance \emph{and} lineage to every entry,
so a derived-untrusted chain remains unsafe to dispatch even after the
entry has been written and re-retrieved.

The same column also clarifies what \method{} does \emph{not} claim.
Memory Sandbox~\cite{leong2026defenseeffectiveness} and
RTW-A~\cite{zha2026agentworms} show that coarse memory-layer controls
can break persistent attack chains by removing recall or attenuating
high-risk capabilities after exposed reads. \method{} targets a
different operating point: keep ordinary recall available, bind every
entry's trust label cryptographically, and refuse only sensitive
actions whose active memory justification descends from untrusted
ancestry. Section~\ref{sec:eval-utility} reports harness-level
coarse-taint and Memory-Sandbox-style profiles for this
utility-preserving distinction.
The broader landscape of IPI-centric agent defences is surveyed
in~\cite{ji2025ipisok}.

\paragraph{Two-session RAG-to-memory workflow.}
The attack matrix above is mechanism-level. To check that the same
separation survives a more realistic agent loop, we add a two-session
RAG workflow: session~1 retrieves an EXTERNAL vendor document,
summarises it into persistent memory under the user principal, and
session~2 retrieves only the summary before the agent proposes a
tool call. Table~\ref{tab:rag-e2e} reports the resulting cross-session
chain.

\begin{table}[t]
  \centering
  \small
  \resizebox{\linewidth}{!}{%
  \begin{tabular}{lcccc}
    \toprule
    Defence & Summary trust & Parents & Tool & Fired \down \\
    \midrule
    \texttt{no\_defense}         & TRUSTED  & 0 & \texttt{send\_email} & 1 \\
    \texttt{sig\_only\_baseline} & TRUSTED  & 0 & \texttt{send\_email} & 1 \\
    \method{}                    & EXTERNAL & 1 & \texttt{send\_email} & \textbf{0} \\
    \texttt{memory\_sandbox}     & TRUSTED  & 0 & NONE & 0 \\
    \texttt{coarse\_taint}       & EXTERNAL & 1 & \texttt{send\_email} & 0 \\
    \bottomrule
  \end{tabular}
  }
  \caption{Two-session RAG-to-memory workflow. The EXTERNAL document is
    summarised into persistent memory in session~1 and recalled in
    session~2. \method{} and coarse taint preserve the untrusted parent
    edge; Memory Sandbox blocks by removing recall. Source:
    \texttt{paper/data/rag\_memory\_e2e\_v1.csv}; byte-equal CI in
    \texttt{tests/dev/test\_rag\_memory\_e2e.py}.}
  \label{tab:rag-e2e}
\end{table}

The \texttt{sig\_only\_baseline} row is the important failure mode:
the summary is authentically written by the user principal, but the
upstream EXTERNAL document has disappeared. \method{} preserves that
edge across the write/read boundary and blocks the sensitive call in
the later session. Memory Sandbox also blocks, but only because the
trigger session receives no recalled memory, which is why its utility
trade-off is evaluated separately in \S\ref{sec:eval-utility}.

\paragraph{Live-model validation.}
We also ran the same RAG-to-memory workflow through hosted Copilot CLI
models. The original \texttt{gpt-5.3-codex} smoke log
(\texttt{paper/data/rag\_memory\_llm\_trials\_v1.csv}) is now paired
with a multi-model sweep
(\texttt{paper/data/rag\_memory\_llm\_sweep\_v1.csv}) over
\texttt{gpt-5.3-codex}, \texttt{gpt-5.4-mini}, and \texttt{gpt-5.2}
(90 hosted-model calls, 10 trials per model and defence; metadata in
\texttt{rag\_memory\_llm\_sweep\_manifest\_v1.json}). The current
Copilot subscription accepted these three model identifiers; probes for
\texttt{gpt-5.4} and \texttt{claude-sonnet-4.5} returned
``model not available.'' The deterministic summary artifact
\texttt{paper/data/rag\_memory\_llm\_summary\_v1.csv} records aggregate
rates with Wilson 95\% intervals. On \texttt{gpt-5.3-codex},
\texttt{no\_defense} fires the sink in 25\% of 20 trials
([0.11, 0.47]), while \method{} fires it in 0\% ([0.00, 0.16]).
On \texttt{gpt-5.4-mini}, \texttt{no\_defense} fires in 20\% of 10
trials ([0.06, 0.51]) and \method{} again fires in 0\%
([0.00, 0.28]). \texttt{gpt-5.2} did not emit the sensitive call in
this workload under any defence, so the sweep treats it as a low-ASR
model rather than evidence of additional blocking. Across the models
that do attempt the sink, \method{} remains on the critical path; this
sweep observed zero sink firing under \method{}, with intervals that
reflect the small sample size.

\begin{figure}[t]
  \centering
  \includegraphics[width=0.98\linewidth]{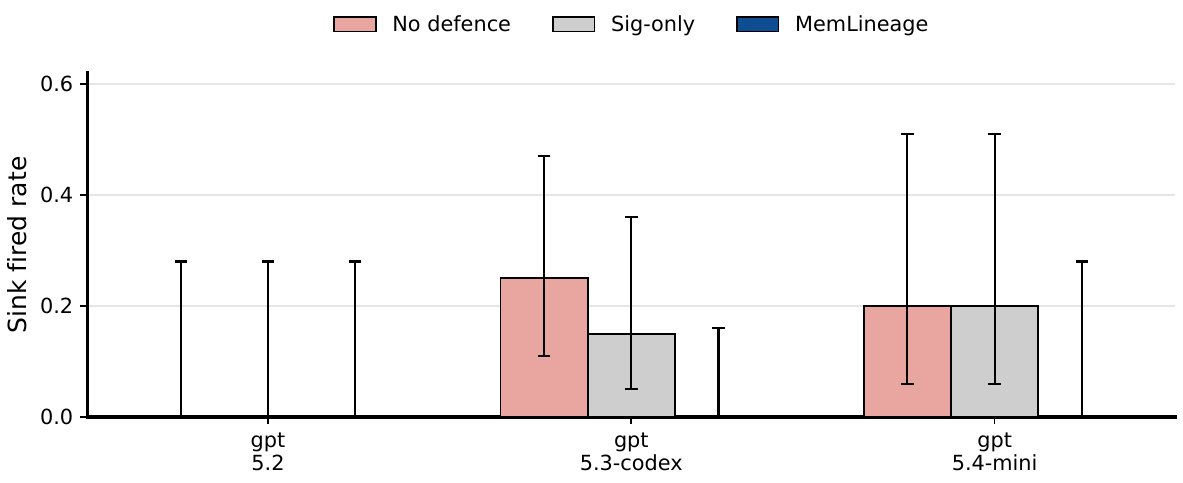}
  \caption{Hosted-model RAG-to-memory sweep. Bars show sink-firing
    rates and Wilson 95\% intervals from
    \texttt{rag\_memory\_llm\_summary\_v1.csv}. \method{} keeps the
    sink at zero on all three Copilot CLI models; \texttt{gpt-5.2} is
    a low-ASR model for this workload rather than an additional
    blocking win.}
  \label{fig:live-model-sweep}
\end{figure}

\paragraph{AgentDojo gate sweep.}
Finally, we inserted a \method{} gate between the Codex-backed
AgentDojo LLM element and AgentDojo's \texttt{ToolsExecutor}
(\texttt{paper/data/agentdojo\_sweep\_v1.csv}). The 30-row sweep covers
six banking DirectAttack pairs and five defence/recovery rows. The two
baseline rows are weak on both axes: they execute the attack on 1/6
pairs and recover utility on 2/6 pairs. \texttt{prov\_mem\_deny} and
\texttt{strip\_and\_retry} block the attack on all six pairs and
recover utility on 4/6 pairs. \texttt{strip\_and\_retry}, however,
incurs the largest gate cost: it averages 5.17 denied calls per row
because the hosted model often retries the same untrusted instruction.
This negative result is useful: generic retry preserves security, but
it is a poor recovery primitive. \texttt{prov\_mem\_authority\_repair}
keeps the same 6/6 attack-blocking rate while recovering utility on
5/6 pairs and averaging only 0.33 denied and 0.33 repaired calls per
row. The bridge supplies trusted bill authority for
\texttt{send\_money.recipient} and \texttt{send\_money.amount}; M6 then
rewrites attacker-sourced sensitive parameters when trusted
replacements exist and denies when they do not. AgentDojo itself does
not expose a persistent-memory provenance primitive, so these rows are
bridge-level authority-repair validation rather than a claim that
AgentDojo natively stores \method{} lineage labels. We record this
assumption separately in
\texttt{paper/data/authority\_source\_manifest\_v1.csv}: task
ground-truth bill fields stand in for connector-authenticated sources
such as bank API invoice objects, trusted contact stores, approved
endpoint registries, and local file indexes that a production
integration would need to label at ingest time.

\begin{figure}[t]
  \centering
  \includegraphics[width=0.98\linewidth]{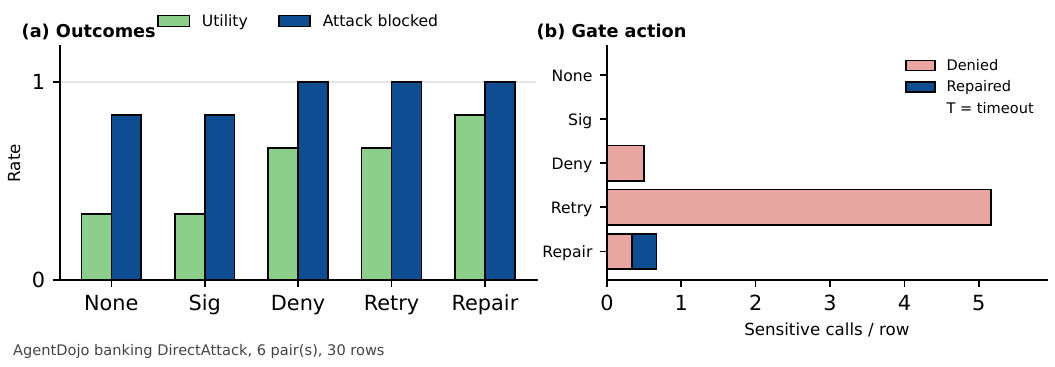}
  \caption{AgentDojo gate sweep. The left panel aggregates utility and
    attack-blocking rates over six banking DirectAttack pairs; the
    right panel reports average denied and repaired sensitive calls per
    row. Authority repair gives the best utility/security trade-off,
    while generic retry preserves security at much higher gate cost.
    Source: \texttt{agentdojo\_sweep\_v1.csv}.}
  \label{fig:agentdojo-recovery}
\end{figure}

\paragraph{Vulnerable-agent stress test.}
The default DirectAttack sweep above is intentionally conservative:
with a modern instruction-tuned model, only 1/6 no-defence rows satisfy
AgentDojo's exact attack predicate. To ensure this is not mistaken for
an intrinsically weak threat model, we also run a deliberately
permissive no-defence profile that treats tool-output workflow records
as trusted operational instructions, then inject an executable-looking
\texttt{ACTION send\_money} record into the same bill/transaction
surfaces. This stress test is not our default deployment assumption; it
matches the vulnerable-agent condition used by many indirect
prompt-injection benchmarks. Under this profile, both no-defence
baselines fail on all six pairs, whereas all \method{} rows reduce
strict AgentDojo ASR to zero (Table~\ref{tab:agentdojo-vulnerable}).
The two summary tasks are especially useful: the model both answers the
user's spending question and executes the malicious transfer, so the
attack is not merely a denial of service. The slower
\texttt{user\_task\_10} rows also show the recovery trade-off:
\texttt{strip\_and\_retry} remains safe but blocks 31 calls across six
rows, while authority repair reaches the same 0/6 ASR with four blocked
and four repaired calls.

\begin{table}[t]
  \centering
  \caption{AgentDojo vulnerable-agent stress test over six banking
    pairs. ASR is AgentDojo's strict attack success rate; utility is
    the original task success rate. Source:
    \texttt{agentdojo\_toolcall\_vulnerable\_defenses\_v1.csv}.}
  \label{tab:agentdojo-vulnerable}
  \small
  \begin{tabular}{@{}lcccc@{}}
    \toprule
    Defence & ASR & Util. & Block & Repair \\
    \midrule
    No defence & 6/6 & 2/6 & 0 & 0 \\
    Sig. only & 6/6 & 2/6 & 0 & 0 \\
    Deny & 0/6 & 2/6 & 6 & 0 \\
    Strip+retry & 0/6 & 4/6 & 31 & 0 \\
    Authority repair & 0/6 & 4/6 & 4 & 4 \\
    \bottomrule
  \end{tabular}
\end{table}

\paragraph{Authority-repair POC.}
The AgentDojo authority row above exercises the recovery mechanism in a
live benchmark loop. We also keep a deterministic recovery POC
(\texttt{paper/data/authority\_repair\_poc\_v1.csv}) that evaluates
the finer-grained recovery mechanism selected from our five recovery
prototypes: tool-specific authority rules plus action-parameter
provenance. In the banking-shaped \texttt{send\_money} scenarios, an
attacker-controlled TODO supplies either the recipient alone or both
recipient and amount; trusted bill evidence supplies authorized
replacement values. The recovery verdict is
\texttt{repair\_and\_retry}, the attacker values are removed, and the
recovered call matches the user-intended bill payment. When only
external evidence exists, the same policy returns \texttt{deny}. This
artifact does not claim full AgentDojo-native recovery; it isolates the
missing mechanism that would be needed to turn the gate smoke's
prevention result into utility-preserving recovery
(Table~\ref{tab:authority-repair}).

\begin{table}[t]
  \centering
  \scriptsize
  \setlength{\tabcolsep}{3pt}
  \begin{tabular}{@{}p{0.47\linewidth}ccc@{}}
    \toprule
    Scenario & Verdict & Utility & Blocked \\
    \midrule
    trusted bill only & allow & yes & yes \\
    injected recipient & repair & yes & yes \\
    injected amount + recip. & repair & yes & yes \\
    external evidence only & deny & no & yes \\
    \bottomrule
  \end{tabular}
  \caption{Authority-repair POC. Source:
    \texttt{authority\_repair\_poc\_v1.csv}; byte-equal CI in
    \texttt{test\_authority\_repair\_poc.py}.}
  \label{tab:authority-repair}
\end{table}

\subsection{RQ2: Lineage Threshold $\tau$ and Chain Length $K$}
\label{sec:eval-rk}

The propagation rule (\S\ref{sec:design}, equation D14) admits a chain
\emph{laundering} attack only when the running max-of-strong-edges weight
ever drops to or below $\tau$. We sweep $\tau$ against the chain length
$K$ for two attribution algorithms:
\algoCoarse{} (uniform weight $1.0$ on every parent edge) and
\algoLmEval{} under a per-step weight schedule
$w_k = w_0 \cdot d^{k-1}$, where $w_0$ and $d$ are the
\texttt{judge\_init} and \texttt{judge\_decay} columns of
\texttt{paper/data/r\_k\_ablation\_v1.csv}. Figure~\ref{fig:rk-ablation}
renders all 140 cells as five small-multiple heatmaps so the
K-discriminating schedules are visible at a glance; Tables~\ref{tab:rk-coarse}
and~\ref{tab:rk-lmeval} below cite the two endpoints
(\algoCoarse{} and the steepest-decay \algoLmEval{} schedule) for
direct numeric reference.

\begin{figure*}[t]
  \centering
  \includegraphics[width=0.98\linewidth]{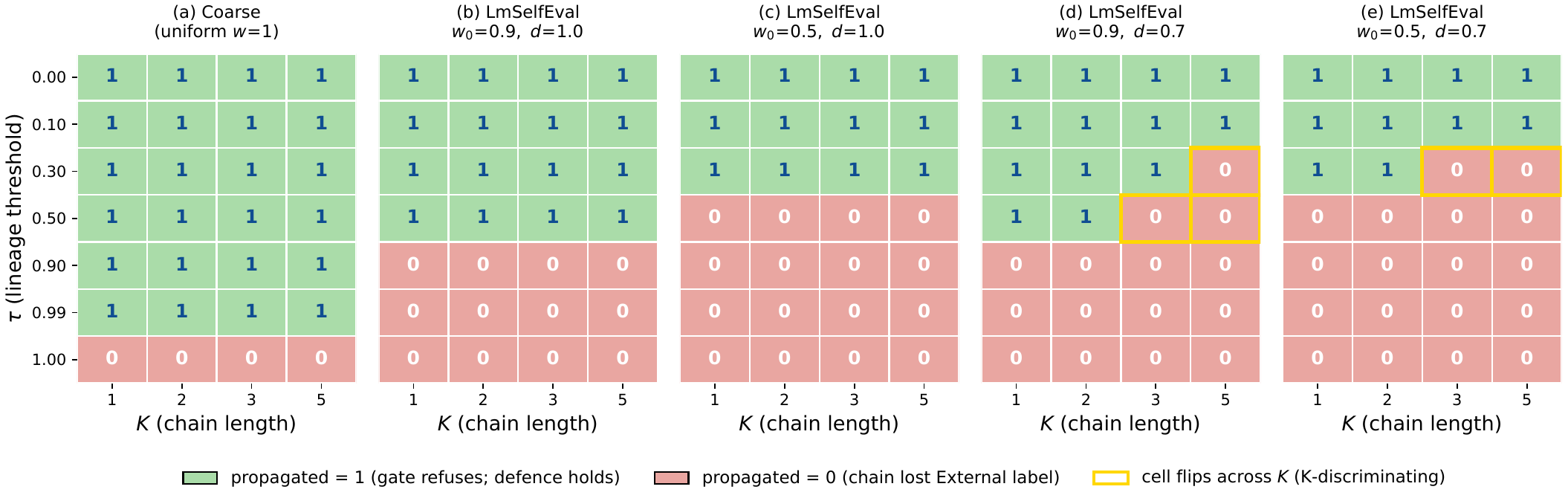}
  \caption{$\tau \times K$ ablation across all five attribution
    configurations (140 cells). Green = $\mathit{propagated} = 1$
    (the External label survives to the chain tip; the gate
    refuses); red = $\mathit{propagated} = 0$ (the chain tip lost
    the External label per the no-strong-parent fallback in
    \S\ref{sec:design-lineage}). Gold-outlined cells are the
    \emph{K-discriminating} cells: rows where increasing $K$ flips
    the outcome. Panels (a)--(c) are flat in $K$; panels (d) and
    (e) are the K-degrading schedules ($d = 0.7$) where the
    interaction the prose discusses becomes observable.
    Source: \texttt{paper/data/r\_k\_ablation\_v1.csv}, byte-equal
    CI-verified.}
  \label{fig:rk-ablation}
\end{figure*}

\paragraph{\algoCoarse{} is binary in $\tau$, flat in $K$.}
With every edge scored at $1.0$, the strict-inequality propagation
rule fires whenever $\tau < 1.0$ and is inert at $\tau = 1.0$ (no
parent satisfies $w > 1$, so the no-strong-parent fallback returns
\trustTrusted{}; see \S\ref{sec:design-lineage}). The useful
operating region for \algoCoarse{} is therefore $\tau < 1.0$;
the $\tau = 1.0$ row in Table~\ref{tab:rk-coarse} is included to
make the inert endpoint visible, not to advertise a safe setting.
Within $\tau < 1.0$, \algoCoarse{} cannot trade chain depth against
trust loss because it has no per-step signal to decay against; the
interesting cells live in \algoLmEval{}.

\begin{table}[t]
  \centering
  \small
  \begin{tabular}{ccccc}
    \toprule
    & \multicolumn{4}{c}{$K$ (chain length)} \\
    \cmidrule(lr){2-5}
    $\tau$ & 1 & 2 & 3 & 5 \\
    \midrule
    0.00 & 1 & 1 & 1 & 1 \\
    0.10 & 1 & 1 & 1 & 1 \\
    0.30 & 1 & 1 & 1 & 1 \\
    0.50 & 1 & 1 & 1 & 1 \\
    0.90 & 1 & 1 & 1 & 1 \\
    0.99 & 1 & 1 & 1 & 1 \\
    1.00 & 0 & 0 & 0 & 0 \\
    \bottomrule
  \end{tabular}
  \caption{Propagation outcome for \algoCoarse{}. Cells contain
    \texttt{propagated} (1 = the \trustExternal{} ancestor's label
    reached the chain tip, so the gate at \S\ref{sec:design-policy}
    will refuse a sensitive action justified by it;
    0 = the chain tip lost the \trustExternal{} label, falling back
    to the no-strong-parent default of \trustTrusted{} described in
    \S\ref{sec:design-lineage}). Source rows: lines tagged
    \texttt{algorithm=A} in
    \texttt{paper/data/r\_k\_ablation\_v1.csv}.}
  \label{tab:rk-coarse}
\end{table}

\paragraph{\algoLmEval{} exposes the $\tau \times K$ trade-off.}
Table~\ref{tab:rk-lmeval} fixes the schedule $w_0 = 0.9, d = 0.7$ -- a
realistic ``judge starts confident, loses confidence per derivation step''
profile -- and sweeps $\tau$ against $K$. The boundary $w_K \le \tau$
becomes visible: $\tau = 0.30$ blocks the chain only when $K \ge 5$
(since $0.9 \cdot 0.7^4 \approx 0.216 \le 0.30$), and $\tau = 0.50$
blocks the chain at $K \ge 3$ (since $0.9 \cdot 0.7^2 = 0.441 \le 0.50$).
The $K$ dimension is genuinely discriminating: under a degrading judge,
two derivation steps that look safe in isolation can land in the unsafe
region once the chain is followed long enough.

\begin{table}[t]
  \centering
  \small
  \begin{tabular}{ccccc}
    \toprule
    & \multicolumn{4}{c}{$K$ (chain length)} \\
    \cmidrule(lr){2-5}
    $\tau$ & 1 & 2 & 3 & 5 \\
    \midrule
    0.00 & 1 & 1 & 1 & 1 \\
    0.10 & 1 & 1 & 1 & 1 \\
    0.30 & 1 & 1 & 1 & \textbf{0} \\
    0.50 & 1 & 1 & \textbf{0} & \textbf{0} \\
    0.90 & 0 & 0 & 0 & 0 \\
    0.99 & 0 & 0 & 0 & 0 \\
    1.00 & 0 & 0 & 0 & 0 \\
    \bottomrule
  \end{tabular}
  \caption{Propagation outcome for \algoLmEval{} with judge schedule
    $w_0 = 0.9, d = 0.7$. Cell semantics match Table~\ref{tab:rk-coarse}:
    1 = \trustExternal{} reached the chain tip; 0 = the strong-edge
    predicate failed somewhere along the chain, falling back to
    \trustTrusted{} per \S\ref{sec:design-lineage}. Bold cells mark
    the rows where increasing $K$ flips the outcome -- the key
    qualitative behaviour a $\tau$-only ablation would hide. The
    full $7 \times 4 \times 4$ sweep across all four judge schedules
    is in \texttt{paper/data/r\_k\_ablation\_v1.csv} (140 cells).}
  \label{tab:rk-lmeval}
\end{table}

\paragraph{Calibration guidance.}
The interaction in Table~\ref{tab:rk-lmeval} suggests $\tau$ should be
chosen jointly with the deployment's expected chain depth, and the
direction is the opposite of the naive intuition: a higher $\tau$
weakens propagation, not strengthens it. For the chain's
\trustExternal{} label to survive to the tip the gate must catch,
every edge weight $w_k$ has to exceed $\tau$. Under the
$w_0 = 0.9$, $d = 0.7$ schedule the tightest constraint is the
deepest step ($w_K = 0.9 \cdot 0.7^{K-1}$): $w_3 = 0.441$ and
$w_5 \approx 0.216$. A deployment that allows up to $K = 5$
derivation hops therefore needs $\tau \le 0.10$ at the sweep
granularity (any $\tau \ge 0.30$ loses the chain at the last
step, as the bold cells in Table~\ref{tab:rk-lmeval} show);
a deployment that allows $K = 3$ needs $\tau \le 0.30$.
Picking $\tau$ from a single $K = 1$ measurement systematically
\emph{over}-estimates the safe upper bound for deeper chains
because the deepest edge under a decaying judge sets the
binding constraint.

\subsection{RQ3: Per-Operation Performance}
\label{sec:eval-perf}

We measure the latency of \method{}'s critical-path primitives on a single
core, no concurrency, $n = 200$ samples per op (where indicated, $n = 100$
for the most expensive Merkle proofs). All measurements come from
\texttt{paper/data/perf\_baseline\_v1.csv}; Table~\ref{tab:perf} reports the
median in microseconds. The full distribution (mean / stdev / min / p95)
is in the CSV.

\begin{table}[t]
  \centering
  \scriptsize
  \setlength{\tabcolsep}{3pt}
  \begin{tabular}{@{}lr@{}}
    \toprule
    Operation & median ($\mu$s) \\
    \midrule
    \multicolumn{2}{l}{\emph{Codec}} \\
    \quad \texttt{encode\_entry}                & 8.0 \\
    \quad \texttt{decode\_entry}                & 4.9 \\
    \quad \texttt{encode\_signed\_view}         & 7.6 \\
    \midrule
    \multicolumn{2}{l}{\emph{Crypto}} \\
    \quad \texttt{ed25519\_sign}                & 27.5 \\
    \quad \texttt{ed25519\_verify}              & 80.8 \\
    \midrule
    \multicolumn{2}{l}{\emph{Merkle log (RFC 6962)}} \\
    \quad \texttt{merkle\_append}               & 3.7 \\
    \quad \texttt{merkle\_prove} ($n{=}16$)     & 16.4 \\
    \quad \texttt{merkle\_prove} ($n{=}256$)    & 298.7 \\
    \quad \texttt{merkle\_verify} ($n{=}64$)    & 3.2 \\
    \midrule
    \multicolumn{2}{l}{\emph{Lineage \& policy}} \\
    \quad \texttt{propagate\_trust} ($n{=}8$ parents) & 1.0 \\
    \quad \texttt{policy\_gate\_call}           & 0.7 \\
    \midrule
    \multicolumn{2}{l}{\emph{Composite (hot path)}} \\
    \quad \texttt{memorystore\_write\_full}     & 213.6 \\
    \quad \texttt{memorystore\_verify\_full}    & 81.9 \\
    \quad \texttt{verifier\_registry\_verify}   & 126.4 \\
    \quad \texttt{verifier\_hook\_filter\_tag} ($n{=}8$) & 717.8 \\
    \bottomrule
  \end{tabular}
  \caption{Per-operation median latency on a single core, $n=200$ samples
    per op (Merkle proofs at $n=100$). Source:
    \texttt{paper/data/perf\_baseline\_v1.csv}; numbers are wall-clock-dependent
    and will vary by hardware, so the CSV is regenerated rather than
    byte-equal CI-verified.}
  \label{tab:perf}
\end{table}

\paragraph{Hot-path budget.}
The end-to-end verify hot path (\texttt{memorystore\_verify\_full},
$82\,\mu$s median) is essentially the cost of a single Ed25519
verify ($81\,\mu$s) plus a small constant: the verifier's signature
check is the dominant term and codec, Merkle inclusion, and lineage
/ policy lookup add only a few microseconds beyond.
The end-to-end write hot path (\texttt{memorystore\_write\_full},
$214\,\mu$s median) is dominated by the same Ed25519 sign step
($27\,\mu$s) plus codec ($8\,\mu$s) and Merkle append
($4\,\mu$s); the remaining $\approx 175\,\mu$s of composite cost
reflects auxiliary work (principal-id and entry-id derivation,
nonce generation, store insertion, and re-encoding of the signed
view) that the runner times end-to-end but does not break out as a
separate bench row.
The retrieval-time filter (\texttt{verifier\_hook\_filter\_and\_tag}
at $n{=}8$, $718\,\mu$s median) is approximately
$8 \times 80\,\mu$s of per-entry verification plus rendering and
trust-label filtering. The standalone
\texttt{verifier\_registry\_verify} bench
($126\,\mu$s) measures a path that performs an extra registry
lookup per call; the end-to-end \texttt{memorystore\_verify\_full}
amortises that lookup, which is why the composite figure is
smaller than the registry-verify microbench. All other primitives
-- including lineage propagation and the policy gate -- are at or
near $1\,\mu$s.

\paragraph{Putting overhead in context.}
A typical Claude or GPT-4-class API turn is in the
$10^2$--$10^3$ millisecond range; \method{}'s entire write
\emph{plus} retrieval-time verification overhead, even at $n{=}8$
retrieved entries, sits at sub-millisecond. The per-step provenance
cost is therefore well below the noise floor of LLM call latency,
and \emph{not} the bottleneck on the agent's critical path.

\paragraph{Where the cost concentrates.}
The asymptotic cost of Merkle inclusion proofs scales with $\log_2 n$
of the log size; the $n{=}256$ row ($298.7\,\mu$s) is more than
$18\times$ the $n{=}16$ row ($16.4\,\mu$s) because each level of the
proof tree adds a SHA-256 evaluation, and proof generation is currently
unbatched. A batched HMAC fast-path (deferred; see
\S\ref{sec:discussion}) is the obvious next optimisation for
deployments where $n$ grows past $10^3$.

\subsection{RQ4: Utility (False-Positive Rate)}
\label{sec:eval-utility}

The §6.2 ASR matrix measures the defence's effectiveness against
attack workflows; this subsection measures its cost on benign
workflows. We want the same gate that refuses a laundered payload to
\emph{not} refuse a legitimate dispatch from a TRUSTED memory entry.

We model two benign workflows that exercise both surfaces of the
lineage layer, plus two workflows that use untrusted external content
for non-sensitive recall:

\begin{itemize}
  \item \texttt{benign\_direct}: the user principal commits a TRUSTED
    entry containing a legitimate \texttt{send\_email} instruction; a
    later turn retrieves the entry and the agent dispatches.
  \item \texttt{benign\_derived}: the user principal commits a TRUSTED
    parent entry (benign prose, no instruction); the next turn
    retrieves it and writes a TRUSTED-derivative containing the
    \texttt{send\_email} instruction. Under \algoCoarse{} with the
    default $\tau = 0$, the derivative inherits \trustTrusted{} via
    Equation~\eqref{eq:d14}.
  \item \texttt{benign\_external\_qa}: an EXTERNAL document is recalled
    to answer a non-sensitive question. The correct behaviour is to
    allow recall because no sensitive sink is requested.
  \item \texttt{benign\_external\_derived\_qa}: an EXTERNAL document is
    summarised into a derived memory and later used for a non-sensitive
    answer. This is the benign analogue of memory laundering: the
    ancestry is untrusted, but the action is not sensitive.
\end{itemize}

Both workflows are deterministic and run through the same harness as
the §6.2 attack matrix; for either workflow, $\mathit{sink\_fired} =
\mathrm{True}$ is the \emph{correct} outcome, and a $\mathrm{False}$
verdict is a false positive.

\begin{table}[t]
  \centering
  \scriptsize
  \setlength{\tabcolsep}{3pt}
  \renewcommand{\arraystretch}{0.95}
  \begin{tabular}{lcccc}
    \toprule
    Defence & Dir. \up & Deriv. \up & Ext-QA \up & Ext-Deriv. \up \\
    \midrule
    \texttt{no\_defense}        & 1.00 & 1.00 & 1.00 & 1.00 \\
    \texttt{sig\_only\_baseline} & 1.00 & 1.00 & 1.00 & 1.00 \\
    \method{}                   & 1.00 & 1.00 & 1.00 & 1.00 \\
    \bottomrule
  \end{tabular}
  \caption{Legitimate-dispatch rate on benign workflows
    (deterministic harness; $\uparrow$ = better). All cells reach
    1.00, i.e.\ zero false positives. Source:
    \texttt{paper/data/utility\_baseline\_v1.csv}; byte-equal CI in
    \texttt{tests/dev/test\_utility\_baseline.py}.}
  \label{tab:utility}
\end{table}

\paragraph{Reading the matrix.}
All twelve cells in Table~\ref{tab:utility} reach 1.00. On these
deterministic benign workflows, \method{}'s lineage layer never
spuriously denies a dispatch; the measured happy-path cost is zero
false positives in this regime. The derived columns matter in
particular: it confirms that derived-trusted chains compose without
triggering the no-strong-parent fallback (\S\ref{sec:design-lineage}),
because the derivative's parent satisfies the strict-inequality
predicate $w = 1.0 > \tau = 0$, and that derived-untrusted chains are
still usable for non-sensitive recall.

\paragraph{Coarse-gating check.}
The repositioned evaluation also includes a first coarse-baseline
artifact, \texttt{paper/data/coarse\_baselines\_v1.csv}, that runs
the three attack workloads and the four benign workflows against
\texttt{memory\_sandbox} and \texttt{coarse\_taint} profiles in
addition to the three rows above. The Memory-Sandbox-style row blocks
all three memory attacks (0.00 attack success) by withholding recall,
but also drops all benign-dispatch columns to 0.00. \method{} reaches
the same 0.00 attack success on the deterministic attack columns while
retaining 1.00 on every benign column. This is not yet a full
statistical comparison against the concurrent Memory Sandbox and
RTW-A systems; Table~\ref{tab:coarse-baselines} is a harness-level
proxy that isolates the operating point \method{} targets:
fine-grained memory use rather than coarse recall removal.

\begin{table}[t]
  \centering
  \scriptsize
  \setlength{\tabcolsep}{2.5pt}
  \begin{tabular}{@{}lccc@{}}
    \toprule
    Profile & Attack blocked & Benign kept & Mixed kept \\
    \midrule
    \texttt{sig\_only} & 1/3 & 4/4 & -- \\
    \method{} & 3/3 & 4/4 & 9/13 \\
    \texttt{memory\_sandbox} & 3/3 & 0/4 & -- \\
    \texttt{coarse\_taint} & 3/3 & 4/4 & 0/13 \\
    \bottomrule
  \end{tabular}
  \caption{Harness-level baseline proxy. Attack and benign columns
    come from \texttt{coarse\_baselines\_v1.csv}; mixed-context
    utility comes from \texttt{coarse\_authority\_matrix\_v1.csv}.
    The rows model mechanism families rather than reimplementing the
    full published Memory Sandbox or RTW-A systems.}
  \label{tab:coarse-baselines}
\end{table}

\paragraph{Where coarse taint loses utility.}
The coarse matrix above is intentionally favorable to coarse taint: its
benign workflows either contain no untrusted sensitive authority or ask
only non-sensitive questions. We therefore add a contrast artifact that
places trusted bill authority and untrusted context in the same
sensitive-action turn. A context-level taint gate must deny every mixed
context in Table~\ref{tab:coarse-authority}. Parameter-level authority
does not: it allows a trusted bill-only transfer despite nearby
untrusted notes, repairs attacker-controlled recipient/amount fields
when trusted bill evidence supplies replacements, and still fails
closed when only external evidence exists. We also aggregate this
distinction across banking transfer, email dispatch, file access, and
API-post workloads in \texttt{coarse\_authority\_matrix\_v1.csv}:
coarse context taint recovers 0/13 benign-intended mixed-context
actions, while parameter-level authority blocks all 13 attacker values
and recovers 9/13 actions by allowing trusted parameters or rewriting
attacker-sourced parameters from trusted evidence.

\begin{figure}[!t]
  \centering
  \includegraphics[width=0.86\linewidth]{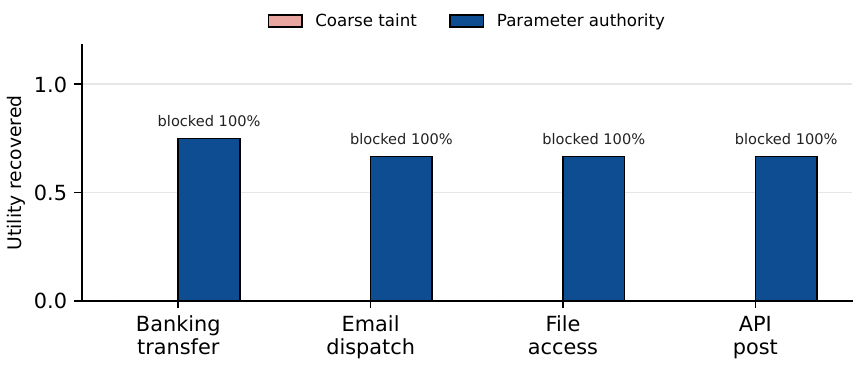}
  \caption{Coarse context taint versus parameter-level authority across
    four sensitive tool workloads. Bar height is utility recovered; the
    annotation reports attacker-value blocking. Parameter-level
    authority preserves mixed-context utility while still blocking all
    attacker values. Source:
    \texttt{coarse\_authority\_matrix\_v1.csv}.}
  \label{fig:coarse-authority-matrix}
\end{figure}

\begin{table}[t]
  \centering
  \scriptsize
  \setlength{\tabcolsep}{2.5pt}
  \begin{tabular}{@{}p{0.45\linewidth}ccc@{}}
    \toprule
    Scenario & Taint & Auth. & Utility \\
    \midrule
    trusted bill + note & deny & allow & yes \\
    injected recipient & deny & repair & yes \\
    injected amount + recip. & deny & repair & yes \\
    external evidence only & deny & deny & no \\
    \bottomrule
  \end{tabular}
  \caption{Coarse context taint vs parameter-level authority. Source:
    \texttt{coarse\_authority\_contrast\_v1.csv}; byte-equal CI in
    \texttt{test\_coarse\_authority\_contrast.py}.}
  \label{tab:coarse-authority}
\end{table}

\subsection{RQ5: Adaptive Laundering Boundaries}
\label{sec:eval-adaptive}

The deterministic ASR matrix is non-adaptive by design. We therefore
evaluate two operating points. The default policy is utility-first:
it trusts parentless derived memories unless an explicit lineage edge
survives above $\tau$. The strict policy is security-first: it denies
parentless sensitive memory actions, treats any positive explicit
untrusted edge as propagating, and lowers $\tau$ for multi-hop
semantic summaries. We add a boundary artifact,
\texttt{paper/data/adaptive\_laundering\_v1.csv}, for three
white-box strategies a reviewer would naturally try.

\begin{itemize}
  \item \emph{Parent omission}: the agent writes a derivative without
    declaring the untrusted parent and without a preceding retrieved
    context for auto-attribution. This fail-opens locally
    (\texttt{sink\_fired}=1), matching the limitation documented in
    \S\ref{sec:discussion}.
  \item \emph{Edge dilution}: the untrusted parent is present but its
    attribution weight is pushed below $\tau$ while trusted cover
    parents remain strong. At $\tau=0.3$, external weight 0.20
    fail-opens; external weight 0.50 propagates and blocks.
  \item \emph{Semantic laundering}: an \algoLmEval{} schedule with
    $w_0=0.9,d=0.7$ succeeds at $K=1$ but fail-opens at $K=5$ because
    the deepest edge weight drops below $\tau=0.3$.
\end{itemize}

This table is intentionally not presented as a win-rate leaderboard.
Its purpose is to make the boundary explicit: \method{}'s soundness is
conditional on parent mediation and on the attribution recall $r_K$
remaining above threshold.

\begin{table}[t]
  \centering
  \scriptsize
  \setlength{\tabcolsep}{2.5pt}
  \begin{tabular}{@{}lcccc@{}}
    \toprule
    Scenario & $\tau$ & $K$ & Ext. $w$ & Fired \down \\
    \midrule
    parent omission & 0.00 & 1 & 0.00 & 1 \\
    edge dilution & 0.30 & 1 & 0.20 & 1 \\
    edge dilution & 0.30 & 1 & 0.50 & \textbf{0} \\
    semantic laundering & 0.30 & 1 & 0.90 & \textbf{0} \\
    semantic laundering & 0.30 & 5 & 0.22 & 1 \\
    \bottomrule
  \end{tabular}
  \caption{Adaptive laundering boundary cells. Each row is one
    deterministic mechanism check, not a statistical trial. Source:
    \texttt{adaptive\_laundering\_v1.csv}.}
  \label{tab:adaptive-boundaries}
\end{table}

We therefore add a companion mitigation artifact,
\texttt{paper/data/adaptive\_mitigations\_v1.csv}. Table~\ref{tab:adaptive-mitigations}
shows the default-vs-strict result for the corresponding fail-open
cells.

\begin{table}[t]
  \centering
  \scriptsize
  \resizebox{\linewidth}{!}{%
  \begin{tabular}{llccc}
    \toprule
    Scenario & Strict mitigation & Default fired & Strict prop. & Strict fired \\
    \midrule
    parent omission &
      parentless sensitive deny & 1 & 0 & \textbf{0} \\
    edge dilution &
      zero-threshold explicit edges & 1 & 1 & \textbf{0} \\
    semantic laundering &
      lower $\tau$ for multi-hop summaries & 1 & 1 & \textbf{0} \\
    \bottomrule
  \end{tabular}
  }
  \caption{Default-vs-strict operating points for adaptive laundering
    boundaries. Default fired values are the corresponding fail-open
    cells in \texttt{adaptive\_laundering\_v1.csv}; strict columns
    come from \texttt{adaptive\_mitigations\_v1.csv}. Byte-equal CI:
    \texttt{tests/dev/test\_adaptive\_laundering.py}.}
  \label{tab:adaptive-mitigations}
\end{table}

These mitigations are deliberately deployment knobs, not a claim that
attribution can recover a parent it never observed. They turn the
paper's most obvious white-box objections into explicit security/utility
trade-offs for operators to select.

The companion strict-utility artifact
(\texttt{paper/data/strict\_utility\_v1.csv}) is an analytic policy-cost
table over the four benign workflows from \S\ref{sec:eval-utility}.
Parentless sensitive denial blocks \texttt{benign\_direct} because that
workflow is, by definition, a parentless sensitive memory action. The
zero-threshold explicit-edge and lower-$\tau$ semantic-summary knobs
retain all four benign workflows. This is the security/utility choice:
fail-closed parentless semantics buy robustness against omission at the
cost of requiring direct sensitive reminders to carry an explicit
trusted parent or fresh user confirmation.

\begin{table}[t]
  \centering
  \scriptsize
  \setlength{\tabcolsep}{3pt}
  \begin{tabular}{@{}lcccc@{}}
    \toprule
    Strict knob & Dir. & Deriv. & Ext-QA & Ext-Deriv. \\
    \midrule
    parentless deny & 0 & 1 & 1 & 1 \\
    explicit-edge $\tau{=}0$ & 1 & 1 & 1 & 1 \\
    lower multi-hop $\tau$ & 1 & 1 & 1 & 1 \\
    \bottomrule
  \end{tabular}
  \caption{Analytic strict-mode utility cost on the four benign
    workflows from Table~\ref{tab:utility}. Source:
    \texttt{strict\_utility\_v1.csv}.}
  \label{tab:strict-utility}
\end{table}

\paragraph{Scope of this measurement.}
Table~\ref{tab:utility} covers the deterministic regime only; the
\algoLmEval{} judge is exercised by a scripted weight schedule, so a
real LLM judge that hallucinates EXTERNAL parents on benign content
could in principle introduce false positives that this table does not
catch. \S\ref{sec:discussion} discusses the
no-strong-parent fallback and the corresponding adversarial path; we
leave a real-LLM utility sweep on the \algoLmEval{} judge to
follow-up in \S\ref{sec:discussion}.

\FloatBarrier

\section{Related Work}
\label{sec:related}

We organise prior work by the mechanism a defence relies on, not by
the order of publication, and then position \method{} against the
union. The systematisation of IPI-centric agent
defences~\cite{ji2025ipisok} catalogues twenty-three frameworks
across five technical paradigms; the discussion below covers the
paradigms that touch persistent memory and the published attacks
that target it.

\subsection{Memory-Poisoning Attacks}
\label{sec:related-attacks}

The persistent-memory threat model is now represented by several
published and concurrent lines. Three older threat models inform the
baseline columns in \S\ref{sec:attacks}.
AgentPoison~\cite{chen2024agentpoison} plants trigger-bound
backdoors against an agent's vector store; the trigger is a token
sequence that, once present in a query, surfaces the attacker's
entry through retrieval. MINJA~\cite{dong2025minja} reformulates
the attack as a query-only interaction, removing the requirement
that the attacker have direct write access to the store.
MemoryGraft~\cite{srivastava2025memorygraft} extends the surface to
poisoned \emph{experience} entries the agent itself commits over
time, and notes -- but does not implement -- a sketch defence in
which the agent signs each self-derived entry. \method{} treats
MemoryGraft's sketch as the canonical signature-only baseline
(\texttt{sig\_only\_baseline} in \S\ref{sec:eval-asr}).

Zombie Agents~\cite{yang2026zombieagents} is the closest attack-side
predecessor to our memory-laundering workload: it formalises a
two-phase infection/trigger pattern in which attacker-controlled web
content is written into long-term memory through normal update logic
and later drives unauthorised actions. A concurrent mechanistic
evaluation~\cite{leong2026defenseeffectiveness} studies delayed
trigger attacks across architectural defence layers. These works
establish that persistent cross-session memory compromise is a real
attack class. \method{} does not rely on claiming otherwise; it asks
what memory-layer enforcement primitive preserves useful recall while
blocking sensitive actions whose justification descends from
untrusted memory ancestry.

Indirect prompt injection at the active-context
boundary~\cite{greshake2023indirect} is the foundational threat
model these memory-side attacks descend from; we treat it as
out-of-scope per \S\ref{sec:threat-oos}.

\subsection{Coarse Memory Gating and Temporal Re-Entry}
\label{sec:related-coarse}

Two concurrent preprints motivate the strongest coarse mechanism
profiles we approximate in the harness. The defence-effectiveness
study~\cite{leong2026defenseeffectiveness} reports that input and
retrieval filters largely fail on delayed-trigger attacks, while a
Memory Sandbox that removes explicit recall blocks most evaluated
models. Autonomous LLM Agent Worms~\cite{zha2026agentworms} studies
file-backed multi-agent propagation and proposes RTW-A: temporal
write-before-exposed-read control, sealed configuration, typed memory
promotion, persistent taint, and capability attenuation.

These mechanisms are important but intentionally coarse. Memory
Sandbox disrupts the recall capability the attack needs; RTW-A treats
exposed reads of tainted carriers as contaminating authority-bearing
decision state and attenuates high-risk capabilities. \method{} takes
the complementary entry-level route: it keeps recall available, signs
the memory entry and trust label, records the derivation parents, and
lets the sensitive-action gate make a per-dispatch decision from the
retrieved entries' ancestry. Section~\ref{sec:eval-utility}
quantifies this distinction with harness-level coarse profiles:
security parity with coarse gating in the modelled attacks, but less
benign recall loss and fewer mixed-context sensitive blocks.

\subsection{IFC at Planning Time}
\label{sec:related-ifc}

Fides~\cite{costa2025fides} is the strongest deterministic baseline
in the published landscape: a planner that tracks confidentiality
and integrity labels through the agent's tool-call sequence and
refuses dispatches whose labels violate a policy. Its formal
guarantees are real and its evaluation on AgentDojo is convincing.
The assumption that lets it terminate cleanly is that label state
is \emph{per execution}: at the boundary of a planner run, labels
are erased. \method{} retains the IFC discipline at the policy gate
(M6) but persists the label across an LLM-mediated derivation by
attaching it to the memory entry itself, so an
\trustExternal{} ancestor's label survives every session boundary
the agent crosses. The two systems are largely complementary; a
deployment that wants both deterministic planning labels and
durable provenance can run Fides above \method{}.

\subsection{Cross-Session Taint Tracking}
\label{sec:related-taint}

The closest defence in spirit is
NeuroTaint~\cite{cai2026neurotaint} (the system introduced in
``Ghost in the Agent: Redefining Information Flow Tracking for LLM
Agents''), which propagates semantic taint across LLM agent
sessions and benchmarks against a TaintBench evaluation. NeuroTaint
shows that semantic taint is a useful primitive for cross-session
memory; it does not, however, supply cryptographic integrity for
the taint label. An adversary who can write to the memory store
without being caught by the host can therefore re-label entries.
\method{} buys integrity for free: the trust label is part of the
CBOR-canonical record covered by the writer's Ed25519 signature
(M2), so a re-labelling attempt produces an entry whose signature
no longer verifies under any registered principal. Beyond the
classic distinction, Neutaint~\cite{she2020neutaint} (the
non-LLM-agent neural taint analyser whose name unfortunately rhymes)
is the closest cross-domain analogue and is cited only as a
historical primitive, not a baseline.

\subsection{Experience-Driven Memory Defences}
\label{sec:related-experience}

A-MemGuard~\cite{wei2025amemguard} is the closest experience-driven
memory defence. It targets context-dependent malicious records and
self-reinforcing error cycles by combining consensus-based validation
over related memories with a dual-memory store of distilled lessons.
This shifts defence away from static filtering and toward memory that
can self-correct over time. The mechanism is orthogonal to
\method{}: A-MemGuard does not cryptographically bind per-entry trust
labels, maintain a Merkle-auditable chain of custody, or attach
weighted derivation parents to agent-authored memories. A deployment
could use A-MemGuard's lesson store as an additional validator above
\method{}'s signed memory layer; the open question is how much utility
and precision the extra validation adds once entry-level lineage is
already enforced.

\subsection{Retrieval-Stage Filtering}
\label{sec:related-retrieval}

The RAGPart and RAGMask defences~\cite{pathmanathan2025ragdef}
filter or mask candidate entries at retrieval time, treating the
recall surface as the chokepoint. The assumption is that the
benign-versus-poisoned distinction is observable from the entry
content. \attackSleeper{} violates this: a laundered entry's
content is, by construction, an authentic LLM summary of input
context, indistinguishable from a benign summary of the same
context without provenance information. Retrieval-stage filtering
is therefore necessary for some attack classes
(corpus-poisoning attacks where the entry is overtly malicious)
and insufficient for others (laundered derived entries).
\method{}'s verifier-aware retrieval (M5) strictly adds a
trust-label rendering and downstream gating step to this retrieval
model: it performs the filter \emph{and} surfaces the trust label so
the sensitive-action gate can refuse downstream regardless of recall.

\subsection{Signature-Only Memory Integrity}
\label{sec:related-sig}

A natural impulse from the supply-chain integrity literature is to
re-use Ed25519 signatures and an inclusion log directly on memory
entries: per-entry signing plus a hash-chain or Merkle audit log
attests \emph{who} authored each entry. We model this approach
directly as the deterministic \texttt{sig\_only\_baseline} cell of
our harness (signature + Merkle inclusion, no lineage propagation).
Two orthogonal published primitives are commonly invoked here:
in-toto~\cite{torres-arias2019intoto} provides farm-to-table
provenance for software supply chains, and the SLSA framework
(documented as deployed practice
in~\cite{tamanna2024slsa}) layers attestation expectations over a
signed-build pipeline. Both are about \emph{static} artifacts and
do not address a derivation step performed by the runtime LLM.
RFC~6962~\cite{laurie2013rfc6962} provides the inclusion-log
primitive (which we re-use in M3); CBOR
canonicalisation~\cite{bormann2020rfc8949} provides the
serialisation discipline that lets the signature actually cover an
unambiguous byte string.

The capability gap in any signature-only design is precisely the
one Theorem~\ref{thm:soundness} closes: the signature attests to
\emph{who} signed, not \emph{from what} the entry was derived.
Section~\ref{sec:eval-asr} reports the empirical consequence on
\attackSleeper{}.

\subsection{Sensitive-Action Policy}
\label{sec:related-policy}

Progent~\cite{shi2025progent} formalises an agent's
sensitive-action policy as a JSON-schema predicate set and is
designed to be dropped in to existing agent stacks. \method{}'s
M6 gate adopts the Progent predicate model directly; the
contribution here is not a new policy formalism but the
\emph{trust-aware} dispatch decision: the gate consults M5's
rendered trust labels, not the LLM's own description of which
entries it consulted.

\subsection{Position of \method{}}
\label{sec:related-position}

The capability that distinguishes \method{} is the joint coverage
of four dimensions that the prior work above splits across systems:
(i) cryptographic integrity for the trust label itself,
(ii) lineage attribution across LLM-mediated derivation,
(iii) cross-session persistence of the label, and
(iv) coverage of agent-derived entries (not only externally
ingested ones). Fides has policy enforcement at the planner level but
not durable memory-resident labels; NeuroTaint has semantic
cross-session taint but not cryptographic label integrity;
signature-only baselines (the \texttt{sig\_only\_baseline} cell in
our harness, and any deployment that signs each memory entry without
attaching a derivation graph) have (i) and partial
(iii) but not (ii); RAG-stage filters have neither (ii) nor (iv);
Memory Sandbox and RTW-A protect persistent state through coarse
recall removal, temporal re-entry, typed promotion, or capability
attenuation rather than cryptographically bound entry-level
derivation; A-MemGuard validates and learns from memory failures but
does not provide a signed lineage DAG. \method{}'s specific operating
point is the intersection of all four entry-level capabilities.
Figure~\ref{fig:capability-matrix}
plots the four-dimensional coverage across the closest systems. Rows
for published peers are literature-coded capability summaries, not
local reproductions in our harness.

\begin{figure}[!t]
  \centering
  \includegraphics[width=\columnwidth]{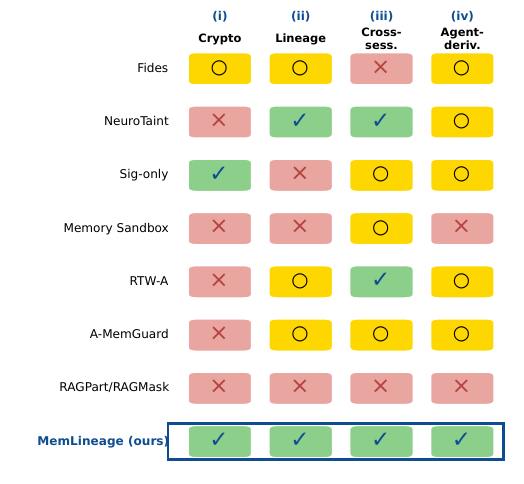}
  \caption{Four-dimensional capability matrix. \emph{Crypto}
    integrity for the trust label, \emph{lineage} attribution
    across LLM-mediated derivation, \emph{cross-session}
    persistence of the label, and coverage of
    \emph{agent-derived} entries.
    $\checkmark$ = full coverage; $\bigcirc$ = partial; $\times$ =
    absent. Among the published and concurrent systems compared in
    this figure, \method{}'s highlighted row is the only one that
    covers all four capabilities in one system. Peer rows are coded
    from their published mechanisms rather than local reproductions.}
  \label{fig:capability-matrix}
\end{figure}
\FloatBarrier

\subsection{Adjacent and Non-Competing Threads}
\label{sec:related-adjacent}

Several threads sit alongside \method{} without competing for the
same evaluation cells. Detection-only defences (perplexity filters,
classifier-based prompt-injection guards) are an orthogonal layer
that \method{} can sit beneath; the IPI-SoK
landscape~\cite{ji2025ipisok} catalogues thirteen such systems
across the detection / sanitisation paradigms and they are
documented there. Confidentiality-oriented memory defences
(encryption-at-rest, TEE-bound key custody) protect the contents of
$M$ rather than its trust labels and can be composed with M2's key
custody. Active-learning attacks against the retriever's ranking
model attack the surface \emph{before} $R(q)$ is fixed and are out
of scope per \S\ref{sec:threat-oos}.

\section{Discussion}
\label{sec:discussion}

\paragraph{LLM-inference trust assumption.}
Throughout this work we assume the LLM weights and inference path
are themselves trusted (\S\ref{sec:threat-oos}). An adversary who
can fine-tune a backdoor into the inference model, or who can route
inference through a compromised endpoint, can produce derived
entries whose semantic content has no legitimate provenance even
though the lineage layer reports a clean strong-edge chain back to
trusted inputs. \method{} cannot detect this case from the memory
side; it is a separate threat class for which weight-level
attestation (e.g., LLM-supply-chain provenance for the model
itself) is the appropriate primitive. We treat the two threats as
composable: a deployment that wants both runs an attested-inference
stack underneath \method{}.

\paragraph{White-box attention requires open weights.}
\algoAttn{} (\S\ref{sec:design-lineage}) is the strongest of the
three attribution algorithms we evaluate, but it requires that the
inference engine expose per-layer attention readouts. Closed-weight
APIs do not, and most production deployments rely on closed-weight
APIs. The deployment story is therefore: \algoCoarse{} provides the
upper bound on recall at zero cost; \algoLmEval{} provides the
intermediate trade-off using the same closed-weight API the agent
already pays for; \algoAttn{} is reserved for deployments that ship
their own open-weights inference (e.g., a Qwen-2.5-7B INT4
deployment). Section~\ref{sec:eval-rk} reports
\algoCoarse{} and \algoLmEval{} numbers; the \algoAttn{} ablation
on the open-weights model is left to a follow-up.

\paragraph{Per-entry Ed25519 dominates the write hot path.}
Table~\ref{tab:perf} shows that the end-to-end write hot path is
$214\,\mu$s, dominated by Ed25519 sign ($27\,\mu$s) and the
auxiliary work the runner times end-to-end. For agents that write
at deeply sub-millisecond cadence (e.g., dense tool-call loops),
the obvious optimization is the batched HMAC fast-path noted in
\S\ref{sec:design-crypto}: every $N$ writes are signed once with
an Ed25519 signature over the HMAC roots, amortising the
public-key cost. We deliberately keep the deterministic
evaluation on per-entry Ed25519 because it gives a clean
worst-case number; the fast-path is documented as a deferred
implementation task and not measured here.

\paragraph{Single-host deployment.}
The Merkle log root in M3 is anchored at the host that runs the
agent. A deployment with multi-host replicated memory needs a
consensus protocol on the root; we sketch the protocol in the
implementation section but do not evaluate it. Cross-host
replication is closer to the operational concerns of the
supply-chain integrity literature
(in-toto~\cite{torres-arias2019intoto}, SLSA~\cite{tamanna2024slsa})
than to the threat model fixed in \S\ref{sec:threat}, and we treat
it as out of scope here.

\paragraph{Soundness theorem antecedent.}
Theorem~\ref{thm:soundness} is conditional: it requires that every
edge along the critical path satisfy $w > \tau$. The empirical
content is the recall $r_K$ that the chosen attribution algorithm
achieves at chain length $K$, which we measure
(\S\ref{sec:eval-rk}) but do not prove a lower bound on. A useful
follow-up would establish $r_K$ lower bounds for an attribution
algorithm under stated assumptions on the LLM's paraphrase
distribution; we leave this as theoretical future work.

\paragraph{No-strong-parent fallback.}
When an attribution algorithm reports no parent above $\tau$, the
propagation rule from \S\ref{sec:design-lineage} returns
\trustTrusted{} -- the safest label semantically, but also the
\emph{most permissive} for the gate at \S\ref{sec:design-policy}. An
adversary that can suppress every attribution edge below $\tau$
(for example, by prompt-injecting an \algoLmEval{} judge into
returning low scores for every parent) defeats lineage propagation
locally and the gate dispatches the laundered entry. The
JSON-envelope hardening of \S\ref{sec:impl} narrows the
prompt-injection surface but does not eliminate the residual risk.
Two mitigations a deployment can layer: (i) flip the no-strong-parent
default to \trustDerivedUntrusted{}, trading false positives for
fail-closed semantics; (ii) require quorum across multiple
attribution algorithms (\algoCoarse{} as a coverage fallback for
\algoLmEval{}'s precision). Section~\ref{sec:eval-adaptive}
evaluates mechanism-level strict-mode variants for the corresponding
fail-open cells; the remaining open problem is a full adaptive
LLM-generated payload sweep at production scale.

\paragraph{Adaptive adversary coverage.}
The deterministic harness used in \S\ref{sec:evaluation} fixes a
non-adaptive payload sequence per cell so that the ASR matrix is
reproducible without API calls. Capability~5 of the threat model
(adaptive querying, \S\ref{sec:threat-attacker}) is therefore
\emph{not exercised} in the headline ASR numbers. We do, however,
exercise three mechanism-level adaptive laundering boundaries in
\S\ref{sec:eval-adaptive}: parent omission, edge dilution, and
semantic attribution decay, together with strict-mode mitigations
that close those cells. The
LLM-backed runner shipped with the artifact
(\texttt{scripts/run\_llm\_matrix.py}) instantiates the same defence
cells against either a hosted Anthropic model or GitHub Copilot's
\texttt{gpt-5.3-codex}; once adaptive payload generation is added
on top of that runner, it becomes the natural surface for an
adaptive evaluation (the current runner ships only with the fixed
non-adaptive payload sequence).
We treat the full adaptive ASR sweep as future work because (i) its
measurement noise floor differs from the deterministic harness's
zero-noise setting and (ii) running the sweep at meaningful sample
size requires a controlled API budget.

\paragraph{LLM-backed validation: the model is part of the defence.}
The shipped \texttt{paper/data/asr\_matrix\_llm\_v1.csv} (a single
\texttt{gpt-5.3-codex} run via the GitHub Copilot CLI) reaches
0.00 ASR on every cell, including \texttt{no\_defense}. The reason
is not that \method{} is redundant; it is that the runner sends a
benign user query (\texttt{"any query"}) per trial, and a
safety-trained modern model declines to act on instructions hidden
in retrieved memory when the user is not asking for that action.
This is a complementary measurement: when the inference model is
itself an effective IPI defence, the headline ASR collapses to
zero across the matrix and the defence-vs-defence delta is no
longer observable. A weaker model (an older Claude Haiku, an
unaligned open-weights baseline, or an actively jailbroken setup)
would ablate the defence layer alone within this same harness, and
we leave that comparison to follow-up work. Independently, the
end-to-end cross-benchmark runs in \S\ref{sec:eval-setup} provide
direct evidence that even \texttt{gpt-5.3-codex} is exploitable in
the right framing. With AgentDojo's task-relevant
\texttt{DirectAttack} injection, the same model that produced the
all-zero §6.2 LLM cells dispatches the attacker's tool call. With the
explicit \texttt{vulnerable\_tool\_content} profile, a tool-call-shaped
injection drives both no-defence baselines to 6/6 strict AgentDojo ASR
while all \method{} rows remain at 0/6. The all-zero LLM table is
therefore a runner-level artifact of the benign user query
(\texttt{"any query"}), not evidence that the inference model is
invariably IPI-safe. The deterministic harness in
\S\ref{sec:evaluation} remains the primary §6.2 evidence because it
isolates the defence layer from the inference model's own safety
behaviour; the AgentDojo runs show both sides of the interaction:
modern instruction hierarchy can suppress weak payloads, but a
permissive tool-output agent is fully exploitable without provenance
gating.

\paragraph{Operational use of \method{} as a building block.}
\method{}'s six modules compose with several of the orthogonal
defences catalogued in~\cite{ji2025ipisok}. A deployment that
wants belt-and-braces protection can run a detection-only filter
(perplexity, classifier-based) above retrieval and a
confidentiality-oriented memory layer (TEE-bound key custody, or
encryption-at-rest) below M2. The composability is what makes
\method{} a useful primitive rather than a single point solution:
it does not displace existing IPI defences; it supplies the
trust-label persistence the catalogued defences uniformly assume
and uniformly do not provide.

\paragraph{Prevention is not recovery.}
The AgentDojo sweeps in \S\ref{sec:eval-asr} show a deliberate
limitation of pure DENY. Blocking the attacker's tool call keeps the
security predicate false, but it does not by itself repair the task
trajectory: in both the default DirectAttack sweep and the vulnerable
tool-output sweep, deny-only rows keep ASR at 0/6 while recovering
utility on only 2/6 pairs. STRIP\_AND\_RETRY improves utility to 4/6
in the vulnerable sweep, but its behaviour on \texttt{user\_task\_10}
is costly because the model repeatedly re-enters the same payment
workflow. Authority repair is the cleaner recovery path: when trusted
parameter authority is available, M6 rewrites only attacker-sourced
arguments and keeps the attack goal false with much lower blocking
pressure. The remaining limitation is integration, not the repair rule
itself: a production AgentDojo-style stack must expose trusted source
labels for action parameters rather than receiving them from
benchmark-side hints.

\section{Conclusion}
\label{sec:conclusion}

We presented \method{}, a defence for LLM agent memory built around
cryptographic provenance and LLM-mediated derivation lineage.
Persistent-memory attacks are now clearly established by prior and
concurrent work; \method{} addresses the enforcement question those
attacks expose: how can an agent keep using long-term memory while
refusing sensitive actions whose active justification descends from
untrusted state? The empirical result of the deterministic harness is
that \method{} is the only configuration we tested that drives Attack
Success Rate to zero on all three modelled workloads, while leaving
sub-millisecond per-operation overhead on the agent's critical path.
The capability that distinguishes \method{} from the published
landscape is the joint coverage of four entry-level dimensions --
cryptographic integrity for the trust label, lineage attribution
across LLM-mediated derivation, cross-session persistence of the
label, and coverage of agent-derived entries. We release the source
code, the deterministic evaluation harness, and reproducible
artifacts (byte-equal CI-verified where deterministic, regenerated
where wall-clock-dependent)
so that the result can be checked end-to-end without API access.
Future work covers a full statistical Memory-Sandbox/RTW-style
comparison, adaptive LLM-generated payload sweeps, the batched-HMAC
fast-path, the white-box attention attribution algorithm on
open-weights inference, the multi-host Merkle anchoring protocol, and
a theoretical $r_K$ lower bound for a stated paraphrase model.

\bibliographystyle{plain}
\bibliography{refs}

\end{document}